\documentclass[11pt,a4paper]{article}

\usepackage{amsmath,amsfonts,amssymb,amsthm}
\usepackage{amsthm}
\usepackage{graphicx,color}
\usepackage{boxedminipage}
\usepackage[linesnumbered, boxed, algosection,noline]{algorithm2e}
\usepackage{framed}
\usepackage{thmtools}
\usepackage{thm-restate}
\usepackage{xspace}
\usepackage{todonotes}
\usepackage{footnote}

\usepackage{vmargin}
\setmarginsrb{1in}{1in}{1in}{1in}{0pt}{0pt}{0pt}{6mm}
\usepackage{todonotes}
\usepackage{footnote}
 \usepackage[pdftex, plainpages = false, pdfpagelabels, 
                 bookmarks=false,
                 bookmarksopen = true,
                 bookmarksnumbered = true,
                 breaklinks = true,
                 linktocpage,
                 pagebackref,
                 colorlinks = true,  
                 linkcolor = blue,
                 urlcolor  = blue,
                 citecolor = red,
                 anchorcolor = green,
                 hyperindex = true,
                 hyperfigures
                 ]{hyperref} 
\usepackage{cleveref}                 
 \usepackage{xifthen}
 \usepackage{tabularx}
\usepackage{tikz} 
 \usetikzlibrary{calc}

\newtheorem{theorem}{Theorem}
\newtheorem{lemma}{Lemma}
\newtheorem{claim}{Claim}[section]
\newtheorem{corollary}{Corollary}
\newtheorem{definition}{Definition}
\newtheorem{observation}{Observation}
\newtheorem{proposition}{Proposition}

\theoremstyle{definition}
\newtheorem{reduction}{Reduction Rule}[section]

\newcommand{\dg}{{\sf d}}
\newcommand{\mad}{{\sf mad}}

\newcommand{\true}{{\sf true}}
\newcommand{\false}{{\sf false}}

\DeclareMathOperator{\operatorClassNP}{NP}
\newcommand{\classNP}{\ensuremath{\operatorClassNP}}

\DeclareMathOperator{\operatorClassFPT}{FPT\xspace}
\newcommand{\classFPT}{\ensuremath{\operatorClassFPT}\xspace}
\DeclareMathOperator{\operatorClassW}{W}
\newcommand{\classW}[1]{\ensuremath{\operatorClassW[#1]}}
\DeclareMathOperator{\operatorClassParaNP}{Para-NP\xspace}
\newcommand{\classParaNP}{\ensuremath{\operatorClassParaNP}\xspace}

\newlength{\RoundedBoxWidth}
\newsavebox{\GrayRoundedBox}
\newenvironment{GrayBox}[1]%
   {\setlength{\RoundedBoxWidth}{.93\textwidth}
    \def\boxheading{#1}
    \begin{lrbox}{\GrayRoundedBox}
       \begin{minipage}{\RoundedBoxWidth}}%
   {   \end{minipage}
    \end{lrbox}
    \begin{center}
    \begin{tikzpicture}%
       \node(Text)[draw=black!20,fill=white,rounded corners,%
             inner sep=2ex,text width=\RoundedBoxWidth]%
             {\usebox{\GrayRoundedBox}};
        \coordinate(x) at (current bounding box.north west);
        \node [draw=white,rectangle,inner sep=3pt,anchor=north west,fill=white] 
        at ($(x)+(6pt,.75em)$) {\boxheading};
    \end{tikzpicture}
    \end{center}}

\newenvironment{defproblemx}[2][]{\noindent\ignorespaces%
                                \FrameSep=6pt%
                                \parindent=0pt%
                \vspace*{-1.5em}
                \ifthenelse{\isempty{#1}}{%
                  \begin{GrayBox}{\textsc{#2}}%
                }{%
                  \begin{GrayBox}{\textsc{#2} parameterized by~{#1}}%
                }
                \begin{tabular*}{\textwidth}{@{\hspace{.1em}} >{\itshape} p{1.8cm} p{0.8\textwidth} @{}}%
            }{
                \end{tabular*}%
                \end{GrayBox}%
                \ignorespacesafterend
            }

\newcommand{\defproblema}[3]{
  \begin{defproblemx}{#1}
    Input:  & #2 \\
    Task: & #3
  \end{defproblemx}
}%

\newcommand{\cO}{\mathcal{O}}
\newcommand{\Oh}{\mathcal{O}}

\newcommand{\eg}{{\sf \ell}_{EG}}
\newcommand{\ad}{{\sf ad}}

\newcommand{\pname}{\textsc}
\newcommand{\ProblemFormat}[1]{\pname{#1}}
\newcommand{\ProblemIndex}[1]{\index{problem!\ProblemFormat{#1}}}
\newcommand{\ProblemName}[1]{\ProblemFormat{#1}\ProblemIndex{#1}{}\xspace}

\newcommand{\probSTP}{\ProblemName{Longest $(s,t)$-Path}}
\newcommand{\probLPMAD}{\ProblemName{Longest Path Above MAD}}
\newcommand{\probLCMAD}{\ProblemName{Longest Cycle Above MAD}}
\newcommand{\probLCAD}{\ProblemName{Longest Cycle Above AD}}

\newcommand{\probLC}{\ProblemName{Longest Cycle}}
 
\newcommand{\HamCycle}{\ProblemName{Hamiltonian Cycle}}

\begin{document}

\title{Longest Cycle above Erd{\H{o}}s--Gallai Bound\thanks{The research leading to these results has received funding from the Research Council of Norway via the project  BWCA
(grant no. 314528), Leonhard Euler International Mathematical Institute in Saint Petersburg (agreement no. 075-15-2019-1620),  and  the Austrian Science Fund (FWF) via project Y1329 (Parameterized Analysis in Artificial Intelligence) .}
}

\author{
Fedor V. Fomin\thanks{
Department of Informatics, University of Bergen, Norway.}\\fedor.fomin@uib.no
\and
Petr A. Golovach\addtocounter{footnote}{-1}\footnotemark{}\\petr.golovach@uib.no
\and
Danil Sagunov\thanks{
    St.\ Petersburg Department of V.A.\ Steklov Institute of Mathematics, Russia} \thanks{JetBrains Research, Saint Petersburg, Russia
}\\danilka.pro@gmail.com
\and 
Kirill Simonov\thanks{Algorithms and Complexity Group, TU Wien, Austria}\\kirillsimonov@gmail.com
}

\date{}

\maketitle

\begin{abstract}

In 1959, Erd{\H{o}}s and Gallai proved that  every graph $G$ with average vertex degree $\ad(G)\geq 2$
contains 
 a cycle of length at least $\ad(G)$.
 We provide an algorithm that for $k\geq 0$ in time 
 $2^{\Oh(k)}\cdot n^{\Oh(1)}$ decides whether a $2$-connected $n$-vertex graph $G$ contains a cycle of length at least $\ad(G)+k$. This resolves an open problem  explicitly mentioned in several papers. 
The main ingredients of our algorithm are new graph-theoretical results interesting on their own.

\end{abstract}

\section{Introduction}\label{sec:intro}
The circumference of a graph is the length of its longest (simple) cycle. 
In 1959,  Erd{\H{o}}s and Gallai~\cite{ErdosG59} gave the following, now classical, lower bound for the circumference of an undirected  graph.

\begin{theorem}[Erd{\H{o}}s and Gallai~\cite{ErdosG59}]\label{thm:EG} 
Every graph with $n$ vertices and more than $\frac{1}{2}(n-1)\ell$ edges ($\ell \geq 2$) contains a cycle of length at least  $\ell +1$. 
\end{theorem}

We provide an algorithmic extension of the Erd{\H{o}}s-Gallai theorem: A fixed-parameter tractable  (\classFPT) algorithm with  parameter  $k$,  that decides whether the circumference of a graph is at least $\ell +k$.
 To state our result formally, we need a few definitions. 
 For an undirected  graph $G$ with $n$ vertices and $m$ edges, we define $\eg(G)=\frac{2m}{n-1}$.
 Then by the Erd{\H{o}}s-Gallai theorem, $G$ always has a cycle of length at least  $\eg(G)$. The parameter $\eg(G)$ is closely related to the \emph{average} degree of $G$,  $\ad(G)=\frac{2m}{n}$.   It is easy to see that for every graph $G$ with at least two vertices, $\eg(G)-1\leq \ad(G)< \eg(G)$. 
 
 The \emph{maximum average degree} $\mad(G)$ is the maximum value of $\ad(H)$ taken over all induced subgraphs $H$ of $G$. Note that $\ad(G) \leq \mad(G)$ and $\mad(G)-\ad(G)$ may be arbitrary large.  By Goldberg~\cite{Goldberg84} (see also~\cite{GalloGT89}), $\mad(G) $ can be computed in polynomial time. By Theorem~\ref{thm:EG}, we have that  if $\ad(G)\geq 2$, then $G$ has a cycle of length at least $\ad(G)$ and, furthermore, if $\mad(G)\geq 2$, then there is a cycle of length at least $\mad(G)$.
Based on this guarantee, we define the following problem.

\defproblema{\probLCMAD}%
{A  graph $G$ on $n$ vertices  and an integer $k\geq 0$.}%
{Decide whether $G$ contains a cycle of length at least $\mad(G)+k$.}

Our main result is that this problem is \classFPT  parameterized by $k$. More precisely, we show the following.

\begin{restatable}{theorem}{main}
\label{thm:mad-fpt}
\probLCMAD can be solved in time $2^{\Oh(k)}\cdot n^{\Oh(1)}$ on   $2$-connected graphs.
\end{restatable}

While Theorem~\ref{thm:mad-fpt} concerns the decision variant of the problem, its proof may be easily adapted to produce a desired cycle  if it exists. We underline this because the standard construction of a long cycle that for every $e
\in E(G)$ invokes the decision algorithm on $G-e$, does not work in our case, as edge deletions decrease the average degree of a graph.

Theorem~\ref{thm:mad-fpt} has several corollaries. 
 The following question was explicitly stated in the literature \cite{FominGLPSZ20,FominGSS22}. For a 
$2$-connected graph $G$  and a nonnegative integer $k$, how difficult is it to decide whether $G$ has a cycle   of length at least $\ad(G)+k$? According to \cite{FominGSS22}, it was not known whether the problem   parameterized by $k$ is 
 \classFPT , \classW1-hard, or \classParaNP.   Even the simplest variant of the question: whether a path of length  $\ad(G)+1$ could be computed in polynomial time, was open. 
Theorem~\ref{thm:mad-fpt} resolves this  question becase $\mad(G)\geq\ad(G)$ for every graph $G$. 

\begin{corollary}\label{cor:ad-cycle}
For a 
$2$-connected graph $G$  and a nonnegative integer $k$,  deciding whether $G$ has a cycle   of length at least $\ad(G)+k$ can be done in time $2^{\Oh(k)}\cdot n^{\Oh(1)}$. 
\end{corollary}

Similarly, we have the following corollary. 
\begin{corollary}\label{cor:eg-cycle}
For a 
$2$-connected graph $G$  and a nonnegative integer $k$,  deciding whether $G$ has a cycle   of length at least $\eg(G)+k$ can be done in time $2^{\Oh(k)}\cdot n^{\Oh(1)}$. 
\end{corollary}

An undirected graph $G$ is \emph{$d$-degenerate} if every subgraph of $G$ has a vertex of degree at most $d$, and the \emph{degeneracy} of $G$ is defined to be the minimum value of $d$ for which $G$ is $d$-degenerate.
Since a graph of degeneracy $d$ has a subgraph $H$ with at least $d \cdot  |V (H)|/2$ edges, 
 we have that 
$d\leq \ad(H)\leq \mad(G)$. Therefore,  Theorem~\ref{thm:mad-fpt} implies the following corollary, which is the main result 
of  \cite{FominGLPSZ20}.

\begin{corollary}[\cite{FominGLPSZ20}]\label{cor:d-cycle}
For a 
$2$-connected graph $G$  of degeneracy $d$,  deciding whether $G$ has a cycle   of length at least $d+k$ can be done in time $2^{\Oh(k)}\cdot n^{\Oh(1)}$. 
\end{corollary}

Theorem~\ref{thm:EG} provides the same lower bound on the number of vertices in a longest path. We consider the \probLPMAD problem that, given a graph $G$ and integer $k$, asks whether $G$ has a path with at least $\mad(G)+k$ vertices.  Observe that a graph $G$ has a path with $\ell$ vertices  if and only if the graph $G'$, obtained by adding to $G$ a universal vertex that is adjacent to every vertex of the original graph, has a cycle with $\ell+1$ vertices.  Because $\mad(G')\geq\mad(G)$, Theorem~\ref{thm:mad-fpt} yields the following. 
\begin{corollary}\label{cor:mad-path}
\probLPMAD can be solved in time $2^{\Oh(k)}\cdot n^{\Oh(1)}$ on  connected graphs.
\end{corollary}

We complement Theorem~\ref{thm:mad-fpt} by observing that the $2$-connectivity condition is crucial for tractability due to the fact that the considered properties are not closed under taking biconnected components. In particular, it may happen that every long cycle of a graph is in a biconnected component of small average degree. This observation  yields the following theorem.  

\begin{restatable}{theorem}{lowerbound}\label{thm:lower}
It is \classNP-complete to decide whether an $n$-vertex connected graph $G$ has a cycle of length at least $\eg(G)+1$.
\end{restatable}

The single-exponential dependence in $k$ of algorithm in Theorem~\ref{thm:mad-fpt} is asymptotically optimal:  it is unlikely that \probLCMAD can be solved in $2^{o(k)}\cdot n^{\Oh(1)}$ time. This immediately follows from the well-known result (see e.g.~\cite[Chapter~14]{cygan2015parameterized}) that existence of an algorithm for \HamCycle with running time $2^{o(n)}$ would refute the \emph{Exponential Time Hypothesis} (ETH) of Impagliazzo, Paturi, and Zane~\cite{ImpagliazzoPZ01}. Thus \probLCMAD cannot be solved in $2^{o(k)}\cdot n^{\Oh(1)}$ time, unless ETH fails.

 \medskip\noindent\textbf{Comparison with the previous work.} Two of the recent articles on the circumference of a graph above guarantee are most relevant to our work. 
 The first is the paper of  Fomin, Golovach, Lokshtanov, Panolan, Saurabh, and Zehavi~\cite{FominGLPSZ20} who gave an algorithm that in time $2^{\Oh(k)}\cdot n^{\Oh(1)}$ for a 
 $2$-connected graph $G$ of degeneracy $d$, decides whether $G$ has a cycle  of length at least $d+k$. 
 In the heart of their algorithm is the following ``rerouting'' argument: If a cycle hits a sufficiently ``dense'' subgraph $H$ of $G$, then this cycle can be rerouted inside $H$ to cover all vertices of $H$.  
 The main obstacle on the way of generalizing the result of Fomin et al.~\cite{FominGLPSZ20} ``beyond'' the average degree was the lack of rerouting arguments in graphs of large average degree.

The rerouting arguments in the proof of Theorem~\ref{thm:mad-fpt} use the structural properties of dense graphs developed in the recent work of Fomin, Golovach, Sagunov, and Simonov~\cite{FominGSS22}  (see 
 \cite{FominGSS21} for the full version) on parameterized complexity of finding a cycle above Dirac's bound. We remind that by the classical theorem of Dirac~\cite{Dirac52}, every 2-connected graph has a cycle of length at least $\min\{2\delta(G),|V(G)|\}$, where $\delta(G)$ is the minimum degree of $G$.
 Fomin et al.  gave an algorithm that in time
 $2^{\Oh (k+|B|)} \cdot n^{\Oh (1)}$
 decides whether a $2$-connected graph $G$ contains a cycle of length at least $\min\{2\delta(G-B), |V(G)|-|B|\}+k$, where $B$ is a given subset of vertices which may have ``small'' degrees.  The result of Fomin et al. \cite{FominGSS21,FominGSS22} is ``orthogonal'' to ours in the following sense: It does not imply Theorem~\ref{thm:mad-fpt}  and Theorem~\ref{thm:mad-fpt} does not imply the theorem from  \cite{FominGSS21}.  However, 
  the tools developed in \cite{FominGSS21}, in particular  the new type of graph decompositions called \emph{Dirac decompositions}, appear to be useful in our case too.

From a more general perspective, our work belongs to a popular subfield of  Parameterized Complexity concerning parameterization above/below specified guarantees. In addition to \cite{FominGSS22,FominGLPSZ20}, the parameterized complexity of paths and cycles above some guarantees was studied in \cite{BezakovaCDF17,DBLP:conf/wg/Jansen0N19}, and 
\cite{FominGLP0Z20a}.

\section{Overview of the proof of  the main result}

Here we outline the critical technical ideas leading to our main result, Theorem~\ref{thm:mad-fpt}. 
We first explain our techniques for the \probLCAD problem. Let us remind that in this problem, the task is to decide whether a graph $G$ has a cycle of length at least $\ad(G)+k$. (The difference with $\mad$ is that we do not take the maximum over all subgraphs.) 

The nucleus of our proof is a novel structural analysis of dense subgraphs in graphs with large average degrees. Informally, we prove that if there is a cycle of length at least $\ad(G) + k$ in $G$, then $G$ contains a dense subgraph $H$ and a long (of length at least $\ad(G) + k$) cycle $C$ that ``revolves'' around $H$ (see Figure~\ref{fig:revolv}). By that, we mean the following. First, the number of times cycle $C$ enters and leaves $H$ is bounded by $\cO(k)$. Second, $C$ contains  at least $\ad(G)-ck$ vertices of $H$ for some constant $c$. Moreover, we need a way stronger ``routing'' property of $H$. Basically for any possible ``points of entry and departure'' of cycle $C$ in $H$, we show that these pairs of vertices could be connected in $H$ by internally vertex-disjoint paths of total length  at least $\ad(G)-ck$. Furthermore, such paths could be found in polynomial time. Then everything boils down to the following problem. For a given subgraph $H$ of $G$, we are looking for at most  $k$ internally vertex-disjoint paths outside $H$ 
of total length $\Omega(k)$, each path starts and ends in $H$.  This task can be done in time $2^{\cO(k)}\cdot n^{\cO(1)}$ by making use of color-coding. Finally, if we find such paths, then we could complete them to a cycle of length at least $\ad(G) + k$ by augmenting them by the paths inside of $H$. 

\begin{figure}[h]
\centering
\scalebox{0.7}{
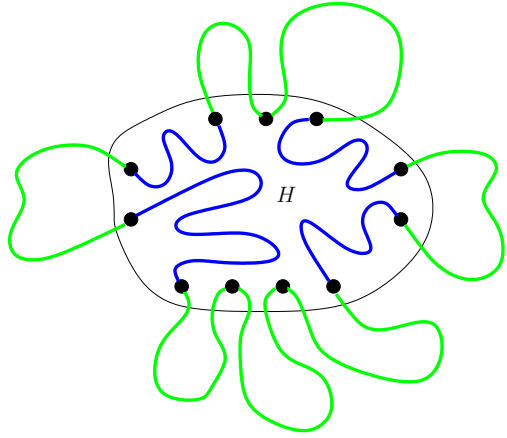}
\caption{A cycle ``revolving'' around $H$. The segments of the cycle outside $H$ are shown in green and the segments inside $H$ are blue.}
\label{fig:revolv}
\end{figure}

\medskip\noindent
\textbf{Identifying dense subgraph $H$.}
Notice that we can assume that $\ad(G)\geq \alpha k$ for a sufficiently big positive constant $\alpha$. Otherwise, we can solve the problem in $2^{\Oh(k)}\cdot n^{\Oh(1)}$ time using the known algorithm for \probLC~\cite{FominLPSZ18}.
We start with preprocessing rules  ``illuminating'' some ``useless'' parts of the graph.  If $G$ contains several connected components, it suffices to keep only the densest of them, as its average degree is at least the average degree of $G$. Similarly, if $G$ is connected but has a cut-vertex, keeping the densest block also suffices. Further, if there is a vertex $v$ of degree less than $ \frac{1}{2} \ad(G)$, then $v$ can be safely removed. By applying these reduction rules exhaustively, we find an induced $2$-connected subgraph $H$ of $G$ whose minimum degree  $\delta(H) \ge \frac{1}{2} \ad(H) \ge \frac{1}{2} \ad(G)$. 
Similarly to removing sparse blocks,  if $G$ contains a vertex separator $X$ of size two such that there is a ``sparse'' component $A$ of $G- X$, then $A$ can be removed.  By applying the last reduction rule we either find a cycle of length at least $\ad(G)+k$ or can conclude that the resulting subgraph $H$ is $3$-connected.

If $(G,k)$ is a yes-instance, that is, graph $G$ contains a cycle of length at least  $\ad(G) +k$, there are two possibilities. Either in   $G$ a cycle of length at least $2\delta(H) + k$ ``lives'' entirely in $H$, or it passes through some other vertices of $G$.  If a long cycle is entirely in $H$, we can employ the recent result of Fomin et al.~\cite{FominGSS21} that finds in time $2^{\Oh(k)} \cdot n^{\Oh(1)}$ in a $2$-connected graph $G$   a cycle of
 length at least $2\delta(G) + k\geq \ad(G) +k$. However, if no long cycle lives entirely in $H$, the result of Fomin et al. is not applicable.  

The next step of constructing $H$ crucially benefits from the graph-theoretical result of Fomin et al.~\cite{FominGSS21}. Specifically, we use the theorem about  the Dirac decomposition from~\cite{FominGSS21}. 
 The definition of the Dirac decomposition is technical and we give it in \Cref{sec:dense}.
  For $2$-connected graphs, the Dirac decomposition imposes a very intricate structure. However, since, thanks to the reduction rules,    $H$ is 3-connected, 
   we bypass most of the technical details from~\cite{FominGSS21}.
Informally,  the Dirac decomposition leads to the following win-win situation.  
By the  Dirac's theorem~\cite{Dirac52}, graph $H$ contains a cycle $S$ of length at least $2 \delta(H) \ge \ad(G)$. Moreover, we could find such a cycle in polynomial time. 
By the result of Fomin et al.~\cite{FominGSS21}, if the length of 
  $S$ is  less than $2\delta(H) + k$,  then either  $S$ can be enlarged in polynomial time, 
  or (a)   $H$ is  small, that is, $|V(H)| < \ad(H) + k$, yielding that $H$ is extremely dense;  or (b) $H$ has a vertex cover of size $\frac{1}{2} \ad(H) - \Oh(k)$. If $S$ got enlarged, we iterate until we achieve cases (a) or (b).
   If we are in case (a), the construction of $H$ is completed. In  case (b), we need to prune the obtained graph a bit more. More specifically, we can delete $\Oh(k)$ vertices in the vertex cover 
   and select a subset of the independent set to achieve the property that (i) each of remaining vertices in  the vertex cover is adjacent to at least  $\ad(H) - \Oh(k)$ vertices in the selected independent subset, and (ii) every vertex of the selected subset of the independent set  sees nearly all vertices of the vertex cover. This mean that the obtained in subgraph  is also ``dense'', albeit in a different sense. Depending on the case, we use different arguments to establish the routing properties of $H$.

\medskip\noindent
\textbf{Routing in $H$.}
 The case (a), when  $|V(H)| < \ad(H) + k$, is easier. In this case,  the degrees of almost all vertices are close to $|V(H)|$. 
 Let $S=\{x_1y_1,\dots, x_\ell y_\ell\}$ an arbitrary set of $\Oh(k)$ pairs of distinct vertices of $H$ forming a linear forest (that is, the union of $x_iy_i$  is a union of disjoint paths). 
 The intuition behind $S$ is that $x_i$ corresponds to the vertex from where the long cycle leaves $H$ and $y_i$ when it enters $H$ again.    
 We show first how to construct 
 a cycle in $H+ S$   (that is, the graph obtained from $H$ by turning the pairs of $S$ into edges) containing every pair $x_iy_i$ from $S$ as an edge. This is done by performing constant-length jumps: any two vertices can be connected either by an edge, or through a common neighbor, or through a sequence of two neighbors.
Then we extend the obtained cycle to a Hamiltonian cycle in $H+S$---every vertex of $H$ that is not yet on a cycle can be inserted due to the high degrees of the vertices. The extension of $S$ into a Hamiltonian cycle is shown in Figure~\ref{fig:appr} (a).

\begin{figure}[ht]
\centering
\scalebox{0.7}{
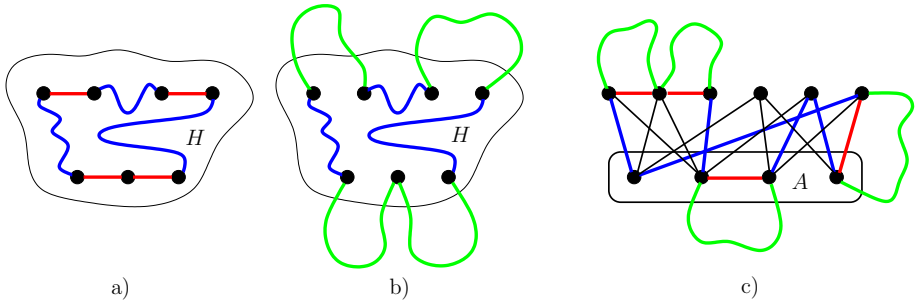}
\caption{Constructing cycles. The set of pairs $S$ that may be both edges and nonedges of $H$ is shown by red lines and the extension of $S$ into a long cycle is blue.  The paths ``revolving'' around $H$ are green. The vertex cover in c) is denoted by $A$.}
\label{fig:appr}
\end{figure}

Therefore, if there is a collection of at most $k$ internally vertex disjoint paths going outside from $H$ and returning back, the high density of $H$ allows collecting all of them in a cycle  containing all the vertices of $H$. Together with all the additional vertices these paths visit outside of $H$ we construct a long cycle in $G$ (see Figure~\ref{fig:appr} (b)). The only condition is that these paths have to form a linear forest. 
Thus, if we find a collection of such paths with enough internal vertices, we immediately obtain a long cycle ``revolving'' around $H$.
The crucial part of the proof is to show that if there is a  \emph{any} cycle of length at least $\ad(H) + k$ in $G$, then it can be assumed to have this form.

Let us remark that a similar ``rerouting'' property was used by Fomin et al.~\cite{FominGLPSZ20} in their above-degeneracy study.
Actually,  for case (a), we need only a minor adjustment of the arguments from ~\cite{FominGLPSZ20}.
However, in the ``bipartite dense'' case (b) the structure of the dense subgraph $H$ is more elaborate and this case requires a new approach.
Contrarily to case (a), the long cycle that we construct in $H+S$ is not  Hamiltonian but visits all the vertices of the vertex cover.
(See Figure~\ref{fig:appr} (c).)
In this case, the behavior of paths depends on which part of $H$ they hit. Because of that, while establishing the routing properties, we have to take into count the difference between paths connecting vertices from the vertex cover, independent set, and both. 
Pushing the ``rerouting'' intuition through, in this case, turns out to be quite challenging.

\medskip\noindent
\textbf{Final steps.}
After finalizing the ``rerouting'' arguments above, it only remains to design an algorithm that checks whether there exists a collection of paths in $G$ that start and end in $H$ and have at least a certain number of internal vertices in total. We do it by a color-coding-style approach. For  case (a), such a subroutine has already been developed in the above-degeneracy case~\cite{FominGLPSZ20}.
On the other hand, for the ``bipartite dense'' case (b) we need to impose an additional restriction on the desired paths, as the length of the final cycle also depends on how the paths' end-vertices are distributed between the two parts and we have to incorporate these kinds of constraints in our path-finding subroutine. 

\medskip
Finally, to solve  \probLCMAD, we use the fact that given a graph $G$, we can find an induced subgraph $F$ with $\ad(F)=\mad(G)$ in polynomial time by the result of Goldberg~\cite{Goldberg84} (see also~\cite{GalloGT89}). Then we find a dense subgraph $H$ of $F$ with the described properties and use $H$ to find a cycle of length at least $\mad(G)+k$.

\section{Preliminaries}\label{sec:prelim} 
In this section, we introduce basic notations, and a series of previously-known results that will be helpful to us.

 We consider only finite undirected graphs.  For  a graph $G$,  $V(G)$ and $E(G)$ denote its vertex and edge sets, respectively. 
 Throughout the paper we use $n=\vert V(G)\vert$ and $m=\vert E(G)\vert$ whenever the considered graph $G$ is clear from the context. 
For a graph $G$ and a subset $X\subseteq V(G)$ of vertices, we write $G[X]$ to denote the subgraph of $G$ induced by $X$. 
We write $G-X$ to denote the graph $G[V(G)\setminus X]$; for a single-element set $X=\{x\}$, we write $G-x$. Similarly, if $Y$ is a set of pairs of distinct vertices, $G-Y=(V(G),E(G)\setminus Y)$. 
For a set $Y$ of pairs of distinct vertices of $G$, $G+Y$ denotes the graph $(V(G), E(G)\cup Y)$, that is, the graph obtained by adding the edges in $Y\setminus E(G)$; slightly abusing notation we may denote the pairs of such a set $Y$ in the same way as edges.  
For a vertex $v$, we denote by $N_G(v)$ the \emph{(open) neighborhood} of $v$, i.e., the set of vertices that are adjacent to $v$ in $G$.
A set of vertices $X$ is a \emph{vertex cover} of $G$ if for every edge $xy$ of $G$, $x\in X$ or $y\in X$.

A \emph{path} $P$ in $G$ is a subgraph of $G$ with $V(P)=\{v_0,\ldots,v_\ell\}$ and $E(P)=\{v_{i-1}v_i\mid1\leq i\leq \ell\}$.  
We write  $v_0v_1\cdots v_\ell$ to denote $P$; the vertices $v_0$ and $v_\ell$ are \emph{end-vertices} of $P$,  the vertices $v_2,\ldots,v_\ell$ are \emph{internal}, and $\ell$ is the \emph{length} of $P$. 
For a path $P$ with end-vertices $s$ and $t$, we say that $P$ is an $(s,t)$-path. 
Two paths $P_1$ and $P_2$ are \emph{internally disjoint} if no internal vertex of one of the paths is a vertex of the other; note that  end-vertices may be the same. 
For two internally disjoint paths $P_1$ and $P_2$ having one common end-vertex, we write $P_1P_2$ to denote the \emph{concatenation} of $P_1$ and $P_2$. 
A graph $F$ is a \emph{linear forest} if every connected component of $F$ is a path. 
Let $S$ be a set of pairs of distinct vertices of $G$; they may be either edges or nonedges.  We say that $S$ is \emph{potentially cyclable} if $(V(G),S)$ is a linear forest. 
A \emph{cycle} is a graph $C$ with $V(C)=\{v_1,\ldots,v_\ell\}$ for $\ell\geq 3$ and $E(C)=\{v_{i-1}v_i\mid1\leq i\leq \ell\}$, where $v_0=v_\ell$. We may write that $C=v_1\cdots v_\ell$. A cycle $C$ (a path $P$, respectively) is \emph{Hamiltonian}  if $V(C)=V(G)$ ($V(P)=V(G)$, respectively).  A graph $G$ is \emph{Hamiltonian} if it has a Hamiltonian cycle. 

A set of vertices $S$ is a \emph{separator} of a connected graph $G$, if $G-S$ is disconnected. 
For a positive integer $k$, $G$ is \emph{$k$-connected} if $|V(G)|> k$ and for every set $S$ of at most $k-1$ vertices, $G-S$ is connected.  
If $S=\{v\}$ is a separator of size one, then $v$ is called a \emph{cut-vertex}. Note, in particular, that a connected graph with at least three vertices is $2$-connected if it has no cut-vertex. 
A \emph{block} of a connected graph with at least two vertices is an inclusion-wise maximal  induced  subgraph without cut-vertices, that is, either a 2-connected graph or $K_2$. 

The \emph{degree} of a vertex $v$ in a graph $G$ is $\dg_G(v)=|N_G(v)|$. The \emph{minimum degree} of $G$ is $\delta(G)=\min\{\dg_G(v)\mid v\in V(G)\}$. For a nonempty set of vertices $X$, the \emph{average degree} of $X$ is 
$\ad_G(X)=\frac{1}{|X|}\sum_{v\in X}\dg_G(v)$, and the \emph{average degree} of $G$ is $\ad(G)=\ad_G(V(G))=\frac{2m}{n}$. The \emph{maximum average degree} is $\mad(G)=\max\{\ad(H)\mid H\text{ is induced subgraph of }G\}$.

The following observation about the circumference lower bound $\eg(G)$ and the average degree of $G$ is useful for us.

\begin{observation}\label{obs:ad}
For every graph $G$ with at least two vertices $\eg(G)-1\leq \ad(G)< \eg(G)$.
\end{observation}

Goldberg~\cite{Goldberg84} proved that, given a graph $G$, an induced subgraph $H$ of maximum \emph{density}, that is, a subgraph with the maximum value $\frac{|E(H)|}{|V(H)|}$, can be found in polynomial time. This result was improved by  
Gallo,  Grigoriadis, and Tarjan~\cite{GalloGT89}. Note that if $H$ is an induced subgraph of maximum density, then $\mad(G)=\ad(H)$.

\begin{proposition}[\cite{GalloGT89}]\label{prop:densest}
An induced subgraph of maximum density of a given graph $G$ can be found in $\Oh(nm\log(n^2/m))$ time.  
\end{proposition}

We use the lower bound on the length of a longest $(s,t)$-path in a $2$-connected graph via the average degree obtained by Fan~\cite{Fan90a}.

\begin{proposition}[{\cite[Theorem~1]{Fan90a}}]\label{prop:Fan}
Let $s$ and $t$ be two distinct vertices in a $2$-connected graph $G$.  Then $G$ has an $(s,t)$-path of length at least $\ad_G(V(G)\setminus \{s,t\})$.
\end{proposition}

Notice that the proof of Proposition~\ref{prop:Fan} in~\cite{Fan90a} is constructive and a required path can be found in polynomial time. 

It is well-known that \probLC can be solved in $2^{\Oh(n)}\cdot n^{\Oh(1)}$ time. The currently best deterministic algorithm is due to Fomin et al.~\cite{FominLPSZ18}. 

\begin{proposition}[\cite{FominLPSZ18}]\label{prop:LC-best}
\probLC can be solved in $4.884^k\cdot n^{\Oh(1)}$ time.
\end{proposition}

The task of \probSTP is, given a graph $G$ with two \emph{terminal} vertices $s$ and $t$, and a positive integer $k$, decide whether $G$ has an $(s,t)$-path with  at least  $k$ vertices. Fomin et al.~\cite{FominLPSZ18} proved that this problem is \classFPT when parameterized by $k$. 

\begin{proposition}[\cite{FominLPSZ18}]\label{prop:st}
\probSTP can be solved in $2^{\Oh(k)}\cdot n^{\Oh(1)}$ time. 
\end{proposition}

\section{Finding a Dense Subgraph}\label{sec:dense}
Here we show that given an instance of \probLCMAD, we can in polynomial time either solve the problem or find a dense induced subgraph of the input graph. 
This part crucially depends on structural and algorithmic results obtained by Fomin et al. in~\cite{FominGSS21}. To describe these results, we have to define the notion of Dirac decomposition introduced in~\cite{FominGSS21} (see Definition~5) even if the only property which we need is that a $3$-connected graph does not admit such a decomposition.
A \emph{leaf-block} of a connected graph having a cut-vertex 
is a block containing exactly one cut-vertex of the original graph.  
A vertex of a leaf-block is \emph{inner} if it is distinct from the cut-vertex in this block. 
  The definition in~\cite{FominGSS21} uses a set $B$ of vertices of small degrees that could be removed from the graph. For our purposes, we adapt the special case of the Dirac's decomposition corresponding to \cite[Definition~5]{FominGSS21} with  $B=\emptyset$.

\begin{definition}[Dirac's decomposition~\cite{FominGSS21}]
Let $G$ be a 2-connected graph and let $C$ be a cycle in $G$ of length at least $2\delta(G)$. Two disjoint paths $P_1$ and $P_2$ in $G$ induce \emph{a Dirac decomposition} for $C$  in $G$ if the following holds.
\begin{itemize}
\item[(i)] The cycle $C$ is of the form $C=P_1 {P'}P_2{P''}$, where  each of the paths ${P'}$ and ${P''}$ has at least $\delta(G)-2$ edges.
\item[(ii)] For  every connected component $H$ of $G-V(P_1  \cup P_2)$,  one of the following holds:
\begin{itemize}
\item $H$ is $2$-connected and the maximum size of a matching in  $G$ between $V(H)$ and $V(P_1)$  is one,  and between $V(H)$ and $V(P_2)$ is also  one;
\item $H$ is not 2-connected and has at least three vertices (i.e., has a cut-vertex),  exactly one vertex of $P_1$ has neighbors in $H$, that is, $|N_{G}(V(H))\cap V(P_1)|=1$, and no inner vertex from a  leaf-block of $H$  has a neighbor in $P_2$;
\item $H$ is not 2-connected and has at least three vertices,  $|N_{G}(V(H))\cap V(P_2)|=1$, and no inner vertex from  a leaf-block of $H$ has a neighbor in $P_1$.			
\end{itemize}
\item[(iii)] There is exactly one connected component $H$ in $G-V(P_1\cup P_2)$ with $V(H)=V(P')\setminus \{s',t'\}$, where $s'$ and $t'$ are the end-vertices of $P'$.
Analogously, there is exactly one connected component $H$ in $G-V(P_1\cup P_2)$ with $V(H)=V(P'')\setminus \{s'',t''\}$, where $s''$ and $t''$ are the end-vertices of $P''$.
\end{itemize}
\end{definition}

Fomin et al.~\cite[Lemma~20]{FominGSS21} proved the following algorithmic 
result.\footnote{We give a simplified variant of~\cite[Lemma~20]{FominGSS21} for $B=\emptyset$.} 

\begin{proposition}[{\cite[Lemma~20]{FominGSS21}}]\label{prop:Dirac}
Let $G$ be a  $2$-connected graph and $k$ be an integer such that  
$0< k \le \frac{1}{24}\delta(G)$ and 
$\delta(G)< \frac{n-k}{2}$.
Then there is an algorithm that, given a cycle $C$ of length less than $2\delta(G)+k$, in polynomial time finds either
\begin{itemize}
\item a longer cycle in $G$, or
\item a vertex cover of $G$ of size at most $\delta(G)+2k$, or
\item two paths $P_1, P_2$ that induce a Dirac decomposition for $C$ in $G$.
\end{itemize}
\end{proposition}

We use the corollary of Proposition~\ref{prop:Dirac} for $3$-connected graphs.

\begin{corollary}\label{cor:Dirac}
Let $G$ be a  $3$-connected graph and $k$ be an integer such that 
$0< k \le \frac{1}{24}\delta(G)$. 
Then there is an algorithm that, given a cycle $C$ of length less than $2\delta(G)+k$, in polynomial time either
\begin{itemize}
\item returns a longer cycle in $G$, or
\item returns a vertex cover of $G$ of size at most $\delta(G)+2k$, or
\item reports that $C$ is Hamiltonian.
\end{itemize}
\end{corollary}

\begin{proof}
To see the claim, observe that by condition (ii) of the definition of a Dirac decomposition, any graph $G$ admitting such a decomposition has a separator of size $2$. Indeed, following the notation from the definition, let $H$ be a connected component of $G-V(P_1  \cup P_2)$. Note that $|V(H)|\geq 3$. 
If  $H$ is $2$-connected, then the maximum size of a matching in  $G$ between $V(H)$ and $V(P_1)$  is one,  and between $V(H)$ and $V(P_2)$ is also  one. Then one can choose an end-vertex of each edge of the matching  between $V(H)$ and $V(P_1)\cup V(P_2)$ in such a way that these two vertices separate a vertex of $H$ and a vertex of  $V(P_1)\cup V(P_2)$.  Suppose that $H$ is not 2-connected and exactly one vertex $u$ of $P_1$ has neighbors in $H$ and no inner vertex from a  leaf-block of $H$ has a neighbor in $P_2$. Then because $G$ is $2$-connected,  $u$ has a neighbor $v$ in a leaf-block $L$ of $H$ distinct from the unique cut-vertex $w$ of $L$. Then $u$ and $w$ form a separator of size $2$ in $G$. The last case from (ii) is symmetric.  

Observe that it can be easily verified whether $C$ is a Hamiltonian cycle. Suppose that this is not the case. 
Then, by the above, the algorithm from Proposition~\ref{prop:Dirac} cannot return a Dirac decomposition. Therefore, if  $\delta(G)< \frac{n-k}{2}$, it either finds a longer cycle or returns a vertex cover of $G$ of size at most $\delta(G)+2k$. If $\delta(G)\geq \frac{n-k}{2}$, then let $k'=\min\{0,n-2\delta(G)\}$. If  $k'=0$, then $2\delta(G)\geq n$. By the theorem of Dirac~\cite{Dirac52}, $G$ is Hamiltonian and, moreover, a Hamiltonian cycle $C'$ can be constructed in polynomial time (see, e.g.,~\cite{locke1985generalization}).  Then we return $C'$. Let $k'>0$. Note that the length of  $C$ does not exceed $2\delta(G)+k'-1$ in this case. Then we apply the algorithm from Proposition~~\ref{prop:Dirac} using $k'$ instead of $k$, which either finds a longer cycle or returns a  vertex cover of size at most $\delta(G)+2k'\leq \delta(G)+2k$. This completes the proof. 
\end{proof}

\begin{lemma}\label{lem:fid-dense}
There is a polynomial-time algorithm that, given an instance  $(G,k)$ of \probLCMAD, where $0<k\leq \frac{1}{80}\mad(G)-1$, either
\begin{itemize}
 \item[(i)] 
 finds a cycle of length at least $\mad(G)+k$ in $G$, or
 \item[(ii)] finds an induced subgraph $H$ of $G$ with $\ad(H)\geq \mad(G)-1$ such that  $\delta(H)\geq \frac{1}{2}\ad(H)$ and $|V(H)|< \ad(H)+k+1$, or
 \item[(iii)] finds an induced subgraph $H$ of $G$ 
 such that there is a partition $\{A,B\}$ of $V(H)$ with the following properties:
 \begin{itemize}
 \item $B$ is an independent set,
 \item $\frac{1}{2}\mad(G)-4k\leq |A|$,
 \item for every $v\in A$, $|N_H(v)\cap B|\geq 2|A|$,
 \item for every $v\in B$, $\dg_H(v)\geq |A|-2k-2$.
 \end{itemize}
  \end{itemize}
\end{lemma}

\begin{proof}
Let $G$ be a graph and let $k\leq  \frac{1}{80}\mad(G)-1$ be a positive integer. First, we apply Proposition~\ref{prop:densest} and find a densest induced subgraph $H$ of $G$.
Then we apply a series of reduction rules to $H$. It is slightly more convenient for us to use $\eg(H)$ as a measure of density. Note that $\eg(H)>\ad(H)=\mad(G)\geq \eg(H)-1$ by Observation~\ref{obs:ad}. Our reduction rules delete some vertices of $H$ without decreasing $\eg(H)$. However, it may happen that the average degree gets smaller, but since we do not decrease $\eg(H)$, the total decrease of the average degree is at most one.  

The first three  rules follow the classical proof of Theorem~\ref{thm:EG}.  

\begin{reduction}\label{red:del-conn}
If $H$ is disconnected, then find a connected component $F$ of $H$ with the maximum value of $\eg(F)$ and set $H:=F$.
\end{reduction}

The following rule is the reason why we switched from the average degree to the Erd{\H{o}}s--Gallai bound.
\begin{reduction}\label{red:del-block}
If $H$ is connected but not $2$-connected, then find a block $F$ of $H$ with maximum value of $\eg(F)$ and set $H:=F$.
\end{reduction}

\begin{reduction}\label{red:del-min}
If $H$  has a vertex $v$ with $\dg_H(v)\leq \frac{1}{2}\eg(H)$, then set $H:=H-v$.
\end{reduction}

The next rule is more complicated.

\begin{reduction}\label{red:del-bicomp}
If $H$ is $2$-connected and 
has a separator $S$ of size two such that there is a component $F$ of $G-S$ with $\ad_H(V(F))\leq\frac{2}{3}\eg(H)$, then delete the vertices of $F$.
\end{reduction}

The Rules~\ref{red:del-conn}--\ref{red:del-bicomp} are applied exhaustively whenever one of them is applicable. 
In the next claim, we show that this does not decrease the density of the graph.

\begin{claim}\label{cl:reduction}
Let $H'$ is the graph obtained by the exhaustive application of Rules~\ref{red:del-conn}--\ref{red:del-bicomp} to $H$.
Then 
$\eg(H')\geq \eg(H)$.
\end{claim}

\begin{proof}[Proof of Claim~\ref{cl:reduction}]
It is sufficient to show the claim for $H'$ obtained by applying either of the rules once. Let $d=\eg(H)$, and we use $n$ and $m$ to denote the number of vertices and edges, respectively, in $H$.
For Rules~\ref{red:del-conn}--\ref{red:del-min}, the proof follows the classical proof of Theorem~\ref{thm:EG}; we provide the arguments here for completeness. 

To see the claim for Rule~\ref{red:del-conn}, assume that $H$ is a disjoint union of $F_1$ and $F_2$. Denote by $n_i$ and $m_i$ the number of vertices and edges, respectively, in $F_i$ for $i\in\{1,2\}$. 
We claim that $\eg(F_1)\geq d$ or $\eg(F_2)\geq d$. To obtain a contradiction, assume that $\eg(F_1)< d$ and $\eg(F_2)< d$. Then $2m_1<d(n_1-1)$ and $2m_2<d(n_2-1)$. We have that
$2m=2m_1+2m_2<d(n_1+n_2-1)- d\leq d(n-1)$ contradicting $\frac{2m}{n-1}=d$. This shows that Rule~\ref{red:del-conn} is safe.

The safety of Rule~\ref{red:del-block} is proved similarly. Suppose that $H$ is connected and let $v$ be a cut-vertex of $H$. Let $\{X,Y\}$ be a separation of $H$ corresponding to $v$, that is, $X\cup Y=V(H)$, $X\cap Y=\{v\}$, and no vertex of $X\setminus Y$ is adjacent to a vertex of $Y\setminus X$. Let $F_1=H[X]$ and $F_2=H[Y]$. As above, we use $n_i$ and $m_i$ to denote the number of vertices and edges, respectively, in $F_i$ for $i\in\{1,2\}$. 
We clam that $\eg(F_1)\geq d$ or $\eg(F_2)\geq d$. The proof is by contradiction. Assume that $\eg(F_1)< d$ and $\eg(F_2)< d$. Then $2m_1<d(n_1-1)$ and $2m_2<d(n_2-1)$. We have that
$2m=2m_1+2m_2<d(n_1+n_2-2)=d(n-1)$. However, this means that $\frac{2m}{n-1}<d$; a contradiction. This proves the claim for Rule~\ref{red:del-block}.

For Rule~\ref{red:del-min}, let $v\in V(H)$ be a vertex with $\dg_H(v)\leq \frac{1}{2}\eg(H)$ and let $H'=H-v$.  Then
\begin{equation*}
\eg(H')=\frac{2m-2\dg_H(v)}{n-2}\geq \frac{2m-d}{n-2}=\frac{2m-2m/(n-1)}{n-2}=\frac{2m}{n-1}=d,
\end{equation*}
as required.  

Finally, we deal with Rule~\ref{red:del-bicomp}.  Suppose that $H$ is $2$-connected and $H$ has a separator $S$ of size two such that there is a component $F$ of $G-S$ with $\ad_H(V(F))\leq\frac{2}{3}\eg(H)$. Let $n_1$ and $m_1$ be the number of vertices and edges in $F$, respectively. We have that  $H'=H-V(F)$. Denote by $n_2$ and $m_2$ the number of vertices and edges, respectively, in $H'$.  We have to prove that $\eg(H')\geq d$. Assume that this is not the case and $\eg(H')<d$. Then $2m_2<d(n_2-1)$.  
Since each vertex of $S$ is adjacent to at most $n_1$ vertices of $V(F)$ in $G$ and $\ad_H(V(F))\leq\frac{2}{3}d$, we have that for the number of edges $m_1'$ of $H[V(F)\cup S]-E(H[S])$, 
\begin{equation*}
2m_1'\leq \frac{2}{3}dn_1+2n_1=dn_1+2n_1-\frac{1}{3}dn_1
\end{equation*}
and, since $m = m_1' + m_2$,
\begin{equation}\label{eq:del-bicomp}
\eg(H)=\frac{2m}{n-1}<\frac{dn_1+2n_1-dn_1/3+d(n_2-1)}{n_1+n_2-1}=d-n_1\frac{d/3-2}{n_1+n_2-1}. 
\end{equation}
As $d\geq \mad(G)>6$, we obtain that $\frac{1}{3}d-2>0$ and
by (\ref{eq:del-bicomp}), $\eg(H)<d$; a contradiction. Therefore, $\eg(H')\geq d$ as required. This concludes the proof.
\end{proof}

For simplicity, let us use the same notation $H$ for the graph obtained by the  exhaustive application of Rules~\ref{red:del-conn}--\ref{red:del-bicomp}. Since the rules do not decrease the value of $\eg(H)$, we have that $\ad(H)\geq \eg(H)-1\geq \mad(G)-1$. 

Because Rules~\ref{red:del-conn} and \ref{red:del-block} are not applicable, we have that $H$ is 2-connected. Suppose that $H$ has a separator $S=\{x,y\}$ of size two.
Let $F_1$ and $F_2$ be two connected components of $H-S$. 
Because of Rule~\ref{red:del-bicomp}, $\ad_H(V(F_i))>\frac{2}{3}\eg(H)$ for $i\in\{1,2\}$. Let $F'_i=H[V(F_i)\cup S]$ for $i\in\{1,2\}$. 
By Proposition~\ref{prop:Fan}, $F_1'$ has an $(x,y)$-path $P_1$ of length at least 
$\frac{2}{3}\eg(H)$. In the same way, $F_2'$ has an $(x,y)$-path $P_2$ of length at least  $\frac{2}{3}\eg(H)$. 
Concatenating these paths we obtain the cycle $C$ whose length is al least $\frac{4}{3}\eg(H)\geq \frac{4}{3}\mad(G)$. Because  $0<k\leq \frac{1}{80}\mad(G)-1$, $C$ is a cycle of length at least $\mad(G)+k$. Then our algorithm returns $C$ and stops as it is required in (i). 

Assume from now on that $H$ has no separator of size two. Because $|V(H)|\geq \ad(H)+1\geq\eg(H)\geq 3$, $H$ is 3-connected.  Because Rule~\ref{red:del-min} is not applicable, $\delta(H)>\frac{1}{2}\eg(H)$. 
Let $k'=\lceil \mad(G)\rceil+k-2\delta(H)\leq k+1$.    Observe that $H$ has a cycle of length at least $\mad(G)+k$ if and only if $H$ has a cycle of length at least $2\delta(H)+k'$. If $k'\leq 0$, then by the theorem of 
 Dirac~\cite{Dirac52}, $G$ has a cycle $C$ of length at least $\min\{|V(H)|,2\delta(H)\}$ and, moreover, $C$ can be constructed in polynomial time (see, e.g.,~\cite{locke1985generalization}). If the length of $C$ is at least 
 $2\delta(H)$, we have that the length of $C$ is at least $\mad(G)+k$ and our algorithm returns $C$ and stops. Otherwise, if the length of $C$ is less than $2\delta(H)$, $C$ is a Hamiltonian cycle in $H$. Thus, we have that 
 $\ad(H)\geq \mad(G)-1$, $\delta(H)>\frac{1}{2}\eg(H)>\frac{1}{2}\ad(H)$ and $|V(H)|<\mad(G)+k\leq\ad(H)+k+1$. This means that $H$ satisfies condition (ii) of the lemma. Then we return $H$ and stop. 
 
 Now we assume that $k'>0$. Recall that $G$ is 3-connected and $k'\leq k+1\leq \frac{1}{80}\mad(G)\leq \frac{1}{24}\delta(H)$.  
This allows us to apply Corollary~\ref{cor:Dirac}. We find an arbitrary cycle in $H$ and apply the algorithm from Corollary~\ref{cor:Dirac} for $H$ and $k'$ iteratively while the algorithm produces a longer cycle. Let $C$ be the cycle of maximum length produced by the algorithm. 

If the length of $C$ is at least $2\delta(H)+k'$, then the length of $C$ is at least $\mad(G)+k$ and we solved the problem. In this case we return $C$ and stop. Assume that the length of $C$ does not exceed  $2\delta(H)+k'-1$.
Suppose that the algorithm constructed a Hamiltonian cycle. This means that $|V(H)|\leq 2\delta(H)+k'-1< \eg(H)+k\leq \ad(H)+k+1$. Since $\ad(H)\geq\mad(G)-1$ and $\delta(H)>\frac{1}{2}\eg(H)>\frac{1}{2}\ad(H)$, $H$ satisfies (ii). Then we return $H$ and stop. It remains to consider the last case when the algorithm from  Corollary~\ref{cor:Dirac} returns  a vertex cover $X$ of $H$ with $|X|\leq \delta(H)+2k'$. 

Because $k'>0$, we have that $\mad(G)+k-2\delta(H)>0$ and, therefore, $\delta(H)<\frac{1}{2}(\ad(H)+k+1)$. Hence, $|X|\leq \frac{1}{2}(\ad(H)+3k+3)$. Consider $B=V(H)\setminus X$. Because $X$ is a vertex cover of $H$, $B$ is an independent set. Let $p=|X|$ and $q=|B|$. We show some properties of $p$ and $q$.

First, we show that $p\geq\frac{1}{2}\ad(H)$, that is, $|X|\geq \frac{1}{2}\ad(H)$.  We have that $|E(H)|\leq \binom{p}{2}+qp$ and
\begin{equation*}
\ad(H)\leq\frac{2\binom{p}{2}+2pq}{p+q}=\frac{2p^2+2pq-p^2-p}{p+q}=2p-\frac{p^2+p}{p+q}\leq 2p. 
\end{equation*}

Next, we show that $q>12p$, i.e., 
$|B|> 12|X|$.
We have that 
\begin{equation*}
\ad(H)\leq \frac{2\binom{p}{2}+2pq}{p+q}=\frac{p(p-1)+2pq}{p+q}<\frac{2p^2+2pq-p^2}{p+q}=2p-\frac{p^2}{p+q}\leq\ad(H)+3k+3-\frac{p^2}{p+q}.
\end{equation*}
Thus, $\frac{p^2}{p+q}<3k+3$ and $(3k+3)q>p^2-(3k+3)p$. Then $q\geq p(\frac{p}{3k+3}-1)$. Recall that $p\geq\frac{1}{2}\ad(H)$ and $k+1\leq\frac{1}{80}\mad(G)\leq\frac{1}{80}(\ad(H)+1)$. Then $\frac{p}{3k+3}>13$ and 
$q>12p$. 

We use the last property and claim that at most $4k-1$ vertices of $X$ have less than $2p$ neighbors in $B$. For the sake of contradiction, assume that this is not the case. Then 
$|E(H)|\leq \binom{p}{2}+4k\cdot 2p+(p-4k)q$ and 
\begin{align*}
\ad(H)\leq&\frac{2\binom{p}{2}+16kp+2(p-4k)q}{p+q}<\frac{2p^2+2pq-8k(q-2p)}{p+q}=2p-\frac{8k(q-2p)}{p+q}\\
\leq &\ad(H)+3k+3-\frac{8k(q-2p)}{p+q}\leq  \ad(H)+6k-\frac{8k(q-2p)}{p+q}.
\end{align*}
Therefore, $6k\geq \frac{8k(q-2p)}{p+q}$ and $11p\geq q$. However, the last inequality contradicts that $q>12p$. This proves our claim.

We use this  property and define $A=\{v\in X\mid |N_H(v)\cap B|\geq 2p\}$. Since $|X\setminus A|\leq 4k-1$ and $|X|\geq \frac{1}{2}\ad(H)\geq \frac{1}{2}\mad(G)-1$, $|A|\geq \frac{1}{2}\mad(G)-4k$.
 Consider $H'=H[A\cup B]$. We have that $\{A,B\}$ is a partition of $V(H')$ with the properties that $B$ is an independent set, $\frac{1}{2}\mad(G)-4k\leq |A|$, $|N_{H'}(v)\cap B|\geq 2p\geq 2|A|$ for all $v\in A$. Also, $\dg_{H'}(v)\geq |A|-2k'\geq |A|-2k-2$ for all $v \in B$ since by construction $\delta(H) \ge |X| - 2k'$. These are exactly the properties that are required in (iii). Then our algorithm returns $H'$. 

To complete the proof of the lemma, we argue that our algorithm is polynomial. For this, note that Rules~\ref{red:del-conn}--\ref{red:del-bicomp} can be applied in polynomial time, because all connected components, blocks, and separators of size two can be listed in polynomial time. Further, the algorithm from Corollary~\ref{cor:Dirac} is polynomial. Since constructing a cycle length at least $\min\{|V(H)|,2\delta(H)\}$ in a 2-connected graph can be done in polynomial time using the proof of Dirac's theorem, we conclude that the overall running time is polynomial. 
\end{proof}

\section{Covering Vertices of Dense Graphs}\label{sec:cover}

In this section, we prove that, given a sufficiently dense graph and a bounded-size set of pairs of distinct vertices $S$ forming a linear forest, we can find a long cycle in $G+S$ containing all edges from $S$. First, we consider the case where there is a small number of vertices in the graph compared to the average degree. Then, we deal with the case where one part in a bipartition of a dense bipartite graph has bounded size.

Recall that for a set $S$ of  pairs of distinct vertices of a graph $G$, we say that $S$ is potentially cyclable if $(V(G),S)$ is a linear forest. 

\begin{lemma}\label{lem:dense}
Let $G$ be a graph and $k$ be an integer such that  (i) $0< k \le \frac{1}{60}\ad(G)$, (ii) $\delta(G)\geq \frac{1}{2}\ad(G)$, and (iii) $\ad(G)+k>n$. Let also $S$ be a potentially cyclable set of at most $k$ pairs of distinct vertices. Then $G+S$ has a Hamiltonian cycle containing every edge of $S$. 
 \end{lemma}

\begin{proof}
Let $G$ be a graph and  let $k$ be an integer satisfying (i)--(iii). Let $d=\ad(G)$.
Using the property that $k$ is small compared to $d$, we upper bound the number of vertices of degree at most $\frac{4}{5}d$.

\begin{claim}\label{cl:ratio-small}
Less than $\frac{1}{12}n$ vertices of $G$ have degree at most $\frac{4}{5}d$.
\end{claim}
 
\begin{proof}[Proof of Claim~\ref{cl:ratio-small}]
 Suppose that at least  $\frac{1}{12}n$ vertices of $G$ have degree at most $\frac{4}{5}d$. Then 
 \begin{equation*}
d\leq \frac{1}{n}\Big(\frac{4}{60}nd+\frac{11}{12}n(n-1)\Big)=\frac{4}{60}d+\frac{11}{12}(n-1)\leq \frac{4}{60}d+\frac{11}{12}(d+k)=\frac{59}{60}d+\frac{11}{12}k
 \end{equation*}
and, therefore, $d\leq 55k$. However, by (i), $60k\leq d$; a contradiction proving the claim. 
\end{proof} 

Denote by $X$ the set of vertices of $G$ whose degrees are at most  $\frac{4}{5}d$.
Let $S$ be a potentially cyclable set of at most $k$ pairs of distinct vertices of $G$ and let $G'=G+S$. 
We show that  $G'$ has a cycle containing the edges of $S$ and the vertices of $X$.

\begin{claim}\label{cl:cycle}
$G'$ has a cycle $C$  containing every edge of $S$ and every vertex of $X$.
\end{claim}

\begin{proof}[Proof of Claim~\ref{cl:cycle}]
Let  $S=\{x_1y_1,\ldots,x_ry_r\}$. Note that some end-vertices of the edges of $S$ may be the same. However, because $S$ forms a linear forest in $G'$, we can assume without loss of generality that it may only happen that $y_{i-1}=x_i$ for some $i\in\{2,\ldots,r\}$. 
We prove that $G'$ has an $(x_1,y_r)$-path $P$ of length at most $5r-4$.
 
The proof is by induction.  We show that for every $i\in\{1,\ldots,r\}$, $G'$ has an $(x_1,y_i)$-path $P_i$ containing $x_1y_1,\ldots,x_iy_i$ and avoiding the end-vertices of $x_{i+1}y_{i+1},\ldots,x_ry_r$ distinct from $y_i$, such that its length is at most $5i-4$. 

The claim is trivial for $i=1$ as we can set $P_1=x_1y_1$. Assume that $i>1$ and $P_{i-1}$ exists. Consider $x_iy_i$. If $y_{i-1}=x_i$, then we just add $x_iy_i$ to the end of $P_{i-1}$, i.e., set $P_{i}=P_{i-1}x_iy_i$. 
Suppose that $y_{i-1}\neq x_i$. If $y_{i-1}x_i\in E(G')$, we set $P_i=P_{i-1}y_{i-1}x_iy_i$. Similarly, if $y_{i-1}$ and $x_i$ have a common neighbor $z\notin U=\{x_1,\ldots,x_r\}\cup \{y_1,\ldots,y_r\}\cup V(P_{i-1})$, we define $P_i=P_{i-1}y_{i-1}zx_iy_i$.  Assume from now on that these are not the cases.

Recall that $\delta(G)\geq \frac{1}{2}d$ and $k\leq \frac{1}{60}d$. Let $W=U \cup X$. We have that 
$|W|\leq 5(i-1)-3+2(r-i+1)+|X|\leq 5k+|X|$ and, using Claim~\ref{cl:ratio-small}, obtain that 
\begin{equation*}
|W|\leq 5k+\frac{1}{12}n\leq 5k+\frac{1}{12}(d+k)< \frac{1}{2}d.
\end{equation*} 
Because $\dg_{G'}(y_{i-1})\geq \frac{1}{2}d$ and $\dg_{G'}(x_{i})\geq \frac{1}{2}d$, we have 
 that $y_{i-1}$ and $x_i$ have neighbors $u$ and $v$, respectively, such that $u,v\notin W$. If $uv\in E(G')$, we define $P_i=P_{i-1}y_{i-1}uvx_iy_i$. Otherwise, observe that $u,v\notin X$ and, therefore, $\dg_{G'}(u)\geq \frac{4}{5}d$ and $\dg_{G'}(v)\geq\frac{4}{5}d$. Because $n<d+k$, $u$ and $v$ have at least $\frac{3}{5}d-k$ common neighbors. Since $|U|\leq 5(i-1)-3+2(r-i+1)\leq 5k$ and $k\leq \frac{1}{60}d$, $u$ and $v$ have a common neighbor $w\notin U$. Hence, we can set $P_i=P_{i-1}y_{i-1}uwvx_iy_i$.   
 
 Observe that in all cases, we constructed $P_i$ from $P_{i-1}$ by appending to the end-vertex $y_{i-1}$ a path of length at  most 5. This means that the length of $P_i$ is at most $5i-4$. This completes the inductive step and the proof of the existence of $P$ with the desired properties.  
 
 Now we apply similar arguments to show that $P$ can be extended to include every vertex of $X$. More precisely, we prove the following. Let $X\setminus V(P)=Z=\{z_1,\ldots,z_s\}$ and let $z_0=y_r$.
 We show that there is an $(x_1,z)$-path $P'$ with $z\in\{z_0,\ldots,z_s\}$ containing $P$ as a subpath that includes every vertex of $Z$ and has length at most $5r-4+4s$.  

We prove by induction that for every $i\in\{0,\ldots,s\}$, $G'$ has an $(x_1,z)$-path $P_i$ containing $P$ as a subpath such that $z\in \{z_0,\ldots,z_i\}\subseteq V(P_i)$ and the length of $P_i$ is at most $5r-4+4i$.

For $i=0$, we set $P_0=P$ and obtain that the claim holds. Let $i\geq 1$ and assume that an $(x_1,z)$-path $P_{i-1}$ with the required properties exists. If $z_i\in V(P_{i-1})$, we take $P_i=P_{i-1}$. Assume that $z_i\notin V(P_{i-1})$. If $zz_i\in E(G')$, we set $P_i=P_{i-1}zz_i$. If $z$ and $z_i$ have a common neighbor $v\notin V(P_{i-1})$, we define $P_i=P_{i-1}zvz_i$. Assume that these are not the cases.

Let $W=V(P_{i-1})\cup X$.
Observe that 
\begin{equation}\label{eq:W}
|W|\leq |V(P)|+4|X|\leq 5k+4|X|\leq 5k+\frac{4}{12}n\leq 5k+\frac{1}{3}(d+k)\leq\frac{4}{9}d,
\end{equation}
by Claim~\ref{cl:ratio-small} and because $k\leq\frac{1}{60}d$. As $\dg_{G'}(z)\geq \frac{1}{2}d$ and $\dg_{G'}(z_i)\geq \frac{1}{2}d$, there are neighbors $u$ and $v$ of $z$ and $z_i$, respectively, such that $u,v\notin W$. 
 If $uv\in E(G')$, we let $P_i=P_{i-1}zuvz_i$.  If $u$ and $v$ are not adjacent, we use the property that $\dg_{G'}(u)\geq \frac{4}{5}d$ and $\dg_{G'}(v)\geq\frac{4}{5}d$, because $u,v\notin X$. 
Then  $u$ and $v$ have at least $\frac{3}{5}d-k$ common neighbors. Note that $|V(P_{i-1})|\leq |W|\leq\frac{4}{9}d$. Then $u$ and $v$ have at least $\frac{7}{45}d-k>0$ common neighbors that are not in $V(P_{i-1})$.  
Let $w$ be such a neighbor. Then we set $P_i=P_{i-1}zuwvz_i$.

 Since  $P_i$ is constructed from $P_{i-1}$ by appending to $z$  a path of length at  most 4,  the length of $P_i$ is at most $5i-4$. This completes the inductive step and we conclude that $P'$ exists.
 
 Now we have that $G'$ has an $(x_1,z)$-path $P'$ of total length length at most $5r-4+4s$ that contains every edge of $S$ and every vertex of $X$. To complete the proof, we show that we can connect the end-vertices of $P'$ to form a cycle. This is trivial if $x_iz\in E(G')$ or if $x_1$ and $z$ have a common neighbor $v\notin V(P')$. Suppose that these are not the cases.  Note that $|V(P')|\leq 5k+ 4|X|$ and, by the same arguments as in (\ref{eq:W}), 
 $|V(P')|\leq\frac{4}{9}d$. This means that $x_1$ and $z$ have neighbors $u$ and $v$, respectively, such that $u,v\notin V(P')$, because $\dg_{G'}(x_1)\geq\frac{1}{2}d$ and $\dg_{G'}(z)\geq\frac{1}{2}d$. If $uv\in E(G')$, we connect the end-vertices of $P'$ by the path $x_1uvz$. Otherwise, we again use the fact that $u,v\notin X$ and, therefore,  $\dg_{G'}(u)\geq \frac{4}{5}d$ and $\dg_{G'}(v)\geq\frac{4}{5}d$. In the same way as above,  
 $u$ and $v$ have at least $\frac{3}{5}d-k$ common neighbors and at least one common neighbor $w\notin V(P')$. Then $P'$ is completed to a cycle by adding the path $x_1uwvz$. This completes the proof.
\end{proof}

By Claim~\ref{cl:cycle},  $G'$ has a cycle $C$  containing every edge of $S$ and every vertex of $X$. Suppose that $C$ is a cycle of this type that has maximum length. We prove that $C$ is Hamiltonian. 

The proof is by contradiction. Assume that $C$ is not Hamiltonian. We consider two cases depending on the length of $C$.  

\medskip
\noindent
{\bf Case~1.} $|V(C)|\leq \frac{1}{2}d$. Consider an arbitrary edge $xy\in E(C)\setminus S$. Note that such an edge exists because $S$ forms a linear forest. We show that we always can extend $C$ by replacing $xy$ by a path. If $x$ and $y$ have a common neighbor $z\notin V(C)$, then we can replace $xy$ by $xzy$. Otherwise, because $\dg_{G'}(x)\geq \frac{1}{2}d$ and $\dg_{G'}(y)\geq \frac{1}{2}d$, $x$ and $y$ have neighbors $u$ and $v$, respectively, such that $u,v\notin V(C)$. If $uv\in E(G)$, then we replace $xy$ by $xuvy$  and extend $C$. If $uv\notin E(G')$,  then we use the fact that $X\subseteq V(C)$ and, therefore, $u,v\notin X$. Then $\dg_{G'}(u)\geq\frac{4}{5}d$ and $\dg_{G'}(v)\geq \frac{4}{5}d$.  Because $|V(G')|<d+k$, $u$ and $v$ have at least $\frac{3}{5}d-k$ common neighbors. Since $|V(C)|\leq \frac{1}{2}d$ and $k\leq\frac{1}{60}d$,  there is a common neighbor $w$ of $u$ and $v$ such that $w\notin V(C)$. Then we replace $xy$ by $xuwvy$ and again extend $C$. Note that the extended cycle contains the edges of $S$ and the vertices of $X$. However, this contradicts the choice of $C$ as a maximum length cycle with this property. 

\medskip
\noindent
{\bf Case~2.} $|V(C)|> \frac{1}{2}d$.  Since $C$ is not Hamiltonian, there is a vertex $v\notin V(C)$. We show that there is an edge $xy\in E(C)\setminus S$ such that both $x$ and $y$ are adjacent to $v$. Suppose that this is not the case and for every $xy\in E(C)\setminus S$, $v$ is not adjacent to at least one end-vertex.  Consider $R=E(C)\setminus S$. Since $1\leq |S|\leq k$,  the edges of $R$ form a linear forest with at least $\frac{1}{2}d-k$ edges.
Each vertex in $V(C)$ covers at most two edges in $R$.
Then our assumption that $v$ is not adjacent to at least one end-vertex of every edge of $R$ implies that $v$ is not adjacent to at least $\frac{1}{2}\big(\frac{1}{2}d-k\big)> \frac{1}{4}d-k$
vertices of $C$. Because $X\subseteq V(C)$, $v\notin X$ and $\dg_{G'}(v)\geq\frac{4}{5}d$. As  $|V(G')|<d+k$, $v$ can have at most $\frac{1}{5}d+k$ nonneighbors. However, $\frac{1}{4}d-k>\frac{1}{5}d+k$, as $k\leq\frac{1}{60}d$; a contradiction. This proves the existence of $xy\in E(C)\setminus S$ such that both and $x$ and $y$ are adjacent to $v$. But then we can extend $C$ by replacing $xy$ by $xvy$ contradicting the choice of $C$. This conclude the case analysis and the proof of the lemma.

\medskip
Let us remark that the proof is, in fact, constructive and can be turned to a polynomial-time procedure that first constructs a cycle $C$ containing every edge of $S$ and every vertex of $X$, and then extends $C$ until we obtain a Hamiltonian cycle.
 \end{proof}

Now we consider dense bipartite graphs. Similarly to Lemma~\ref{lem:dense}, we show that for a given set of pairs of vertices forming a linear forest there is a cycle containing all these pairs in the extended graph, and also each vertex of the ``high-degree'' part of the graph. For an example, see Figure~\ref{fig:L3}.

\begin{figure}[ht]
\centering
\scalebox{0.7}{
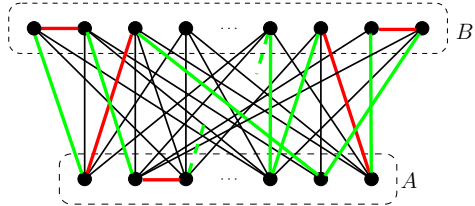}
\caption{Structure of $G$ and $G'=G+S$. The set of pairs $S$ is shown by red lines and the edges of $C$ that are not in $S$ are green. Note that $G'$ is not required to be bipartite.}
\label{fig:L3}
\end{figure}

\begin{lemma}\label{lem:dense-bip}
Let $G$ be a bipartite graph, $\{A,B\}$ is a bipartition of $V(G)$ with $p=|A|$, and let $k$ be an integer such that  (i) $0< k \le \frac{1}{10}p$, (ii) for every $v\in A$, $\dg_G(v)\geq 2p$, and (iii) for every $v\in B$, $\dg_G(v)\geq p-k$. Let $S$ be a potentially cyclable set of at most $\frac{9}{4}k$ pairs of distinct vertices. Then $G'=G+S$ has a cycle $C$ containing every edge of $S$ and every vertex of $A$. Furthermore, $C$ is a  longest cycle in $G'$ containing the edges of $S$ and the length of $C$ is $2p-s+t$, where $s$ is the number of edges of $S$ with both end-vertices in $A$ and $t$ is the number of edges in $S$ with both end-vertices in $B$. 
 \end{lemma}

\begin{proof}
The proof of the lemma follows the same strategy as the proof of Lemma~\ref{lem:dense}. Suppose that $G$, $k$, and $S$ satisfy the conditions of the lemma. Let also $G'=G+S$. Denote by $s$ the number of edges of $S$ with both end-vertices in $A$, and let $t$ be the number of edges in $S$ with both end-vertices in $B$.

\begin{claim}\label{cl:cycle-bip}
$G'$ has a cycle $C$  containing every edge of $S$. 
\end{claim}

\begin{proof}[Proof of Claim~\ref{cl:cycle-bip}]
Let  $S=\{x_1y_1,\ldots,x_ry_r\}$. We can assume without loss of generality that it may only happen that $y_{i-1}=x_i$ for some values $i\in\{2,\ldots,r\}$ and other end-vertices of the edges of $S$ are distinct. 
We prove  that $G'$ has an $(x_1,y_r)$-path $P$ of length at most $5r-4$ containing every edge of $S$.
 
We show inductively that for every $i\in\{1,\ldots,r\}$, $G'$ has an $(x_1,y_i)$-path $P_i$ containing $x_1y_1,\ldots,x_iy_i$ and avoiding the end-vertices of $x_{i+1}y_{i+1},\ldots,x_ry_r$ distinct from $y_{i}$ (it may happen that $x_{i+1}=y_i$) whose length is at most $5i-4$. 

If $i=1$, then we  set $P_1=x_1y_1$ and the claim holds. Assume that $i>1$ and $P_{i-1}$ exists. Consider $x_iy_i$. If $y_{i-1}=x_i$, then we just add $x_iy_i$ to the end of $P_{i-1}$, i.e., set $P_{i}=P_{i-1}x_iy_i$. 
Suppose that $y_{i-1}\neq x_i$. If $y_{i-1}x_i\in E(G')$, we set $P_i=P_{i-1}y_{i-1}x_iy_i$. 
Assume from now on that these are not the cases. Let $U=V(P_{i-1})\cup \{x_1,\ldots,x_r\}\cup\{y_1,\ldots,y_r\}$. Denote $U_A=U\cap A$ and $U_B=U\cap B$.
Observe that $|U_A|\leq \lceil\frac{1}{2}(5r-4)+s\rceil\leq \frac{63}{8}k<8k$. Symmetrically, $|U_B|<8k$.
We consider the following four cases depending on whether $y_{i-1}$ and $x_i$ belong to  $A$ or $B$. 

\medskip
\noindent 
{\bf Case~1.} $y_{i-1},x_i\in B$. Becase $\dg_G(y_{i-1})\geq p-k$ and $\dg_G(x_{i})\geq p-k$, 
$y_{i-1}$ and $x_i$ have at least $p-2k$ common neighbors in $A$.  Because $|U_A|<8k$ and $p\geq 10k$, we obtain that $y_{i-1}$ and $x_i$ have a common neighbor $v\notin U_A$.   Then we 
construct $P_i=P_{i-1}y_{i-1}vx_iy_i$.

\medskip
\noindent 
{\bf Case~2.} $y_{i-1}\in A$ and $x_i\in B$. Because $\dg_{G}(y_{i-1})\geq 2p$, $|U_B|<8k$ and $p\geq 10k$, $y_{i-1}$ has a neighbor $u\in B$ such that $u\notin U_B$. Then applying for $u$ and $x_i$ the same arguments as in Case~1, we obtain that $u$ and $x_i$ have a common neighbor $v\in A$ such that $v\notin U_A$. Then we set $P_i=P_{i-1}y_{i-1}uvx_iy_i$.

\medskip
\noindent 
{\bf Case~3.} $y_{i-1}\in B$ and $x_i\in A$. This case is symmetric to Case~2. Using the same arguments we obtain that $x_i$ has a neighbor $v\in B\setminus U_B$, and $y_{i-1}$ and $v$ have a common neighbor $u\in A\setminus U_A$. Then  $P_i=P_{i-1}y_{i-1}uvx_iy_i$.

\medskip
\noindent 
{\bf Case~4.} $y_{i-1}, x_i\in A$. Because $\dg_{G}(y_{i-1})\geq 2p$, $\dg_{G}(x_{i})\geq 2p$, $|U_B|<8k$ and $p\geq 10k$, $y_{i-1}$ and $x_i$ have neighbors in $B\setminus U_B$. If these vertices have a common neighbor $v$ of this type, then we set $P_i=P_{i-1}y_{i-1}vx_iy_i$. Otherwise, let $u$ and $v$ be neighbors of $y_{i-1}$ and $x_i$, respectively, in  $B\setminus U_B$. Using the arguments from Case~1, we have that $u$ and $v$ have a common neighbor $w\in A\setminus U_A$. Then we define $P_i=P_{i-1}y_{i-1}uwvx_iy_i$.

\medskip
In all cases, $P_i$ was constructed from $P_{i-1}$ by extending it by a path of length at most 5. This competes the inductive step and proves the existence of $P$. 

To complete the proof, we show that the end-vertices of $P$ can be connected by a path $Q$ to form a cycle. This is trivial if $x_1y_{r}\in E(G')$. Otherwise, we construct $Q$ using the same arguments as in above Cases~1--4. Let $U_A=V(P)\cap A$ and $U_B=V(P)$. Because $|V(P)|\leq 5r-3$, we have that $|U_A|<8k$ and $|U_B|<8k$. If $x_1,y_r\in B$, we find a common neighbor $v\in A\setminus U_A$ in the same way as in Case~1 and define $Q=x_1vy_{r}$. If $x_1\in A$ and $y_r\in B$, we find a neighbor $u$ of $x_1$ in $B\setminus U_B$ and then a common neighbor $v$ of $u$ and $y_r$ in $A\setminus U_A$  following the arguments from Case~2. Then $Q=x_1uvy_r$. Then case $x_1\in B$ and $y_r\in A$ is symmetric. Finally, if $x_1,y_r\in B$, we use the same arguments as in Case~4. We either find a common neighbor $v\in B\setminus U_B$ of $x_1$ and $y_r$ and define $Q=x_1vy_r$ or we find two distinct neighbors $u$ and $v$ of $x_1$ and $y_r$, respectively, where $u,v\in B\setminus U_B$. In the last case, we find a common neighbor $w$ of $u$ and $v$ in $A\setminus U_A$, and set $Q=x_1uwvy_r$. This completes the proof. 
\end{proof}

By Claim~\ref{cl:cycle-bip},  $G'$ has a cycle $C$  containing every edge of $S$. Let $C$ be a cycle in $G'$ containing the edges of $S$ that has maximum length. We show that $C$ contains every vertex of $A$. 

The proof is by contradiction. Assume that $A \setminus V(C)\neq\emptyset$. Because $|S|\leq\frac{9}{4}k$, $|V(C)\cap B|\leq p+\frac{9}{4}k$. Then because $k \leq \frac{1}{10}p$ and $\dg_G(v)\geq 2p$ for every $v\in A$, 
$W=B\setminus V(C)\neq\emptyset$ and, moreover,
every vertex $v\in A$ has a neighbor $u\in W$. In fact, every vertex $v\in A$ has at least two distinct neighbors $u\in W$.
We consider two cases depending on the number of vertices of $A$ outside $C$.  

\medskip
\noindent
{\bf Case~1.} $|A\setminus V(C)|>2k$.  Let $xy\in E(C)\setminus S$. We assume without loss of generality that $x\in A$ and $y\in B$. 
 We show that $C$ can be extended by replacing $xy$ by a path. We have that $x$ has a neighbor $u\in W$. 
 Because $\dg_G(u)\geq p-k$ and $\dg_G(y)\geq p-k$, $u$ and $y$ have at least $p-2k$ common neighbors in $A$. Since $|A\setminus V(C)|>2k$, $u$ and $y$ have a common neighbor $v\in A\setminus V(C)$. 
This means that we can replace $xy$ by $xuvy$ and extend $C$. 

\medskip
\noindent
{\bf Case~2.} $|A\setminus V(C)|\leq 2k$. Denote by $R$ the set of pairs $\{x,y\}$ of distinct vertices of $V(C)\cap A$ such that $C$ contains a segment $xvy$ for some $v\in B$ and $xv,yv\notin S$. Note that because
$|A\cap V(C)|\geq p-2k$, $|R|\geq p-2k-|S|\geq \frac{23}{4}k$. Observe also that the pairs of $R$ form a linear forest. Then there is a subset $R'\subseteq R$ of disjoint pairs with $|R'|\geq\frac{1}{2}|R|\geq \frac{23}{8}k>2k$.  

Let $u\in A\setminus V(C)$. Recall that $u$ has two distinct neighbors $v,w\in W$. We claim that there is a pair $\{x,y\}\in R'$ such that $xv,yw\in E(G)$ or $xw,yv\in E(G)$. Because $\dg_G(v)\geq p-k$, by the pigeonhole principle, there are at most $k$ pairs $\{x,y\}\in R'$ such that $xu\notin E(G)$ or $yu\notin E(G)$. Thus, there is $R''\subseteq R'$ of size at least $|R'|-k>k$ such that $xu,yu\in E(G)$ for every $\{x,y\}\in R''$. Since 
$\dg_G(w)\geq p-k$, there are at most $\frac{1}{2}k$ pairs $\{x,y\}\in R''$ such that $xw,yw\notin E(G)$. As $|R''|>k$, we conclude that there is a pair $\{x,y\}\in R''$ such that $xw\in E(G)$ or $yw\in E(G)$. Thus, 
 $xv,yw\in E(G)$ or $xw,yv\in E(G)$. Let $xzy$ be the segment of $C$. 
If   $xv,yw\in E(G)$, we replace $xzy$ by $xvuwy$, and $xzy$ is replaced by $xwuvy$ if $xw,yv\in E(G)$. In both cases, we extend $C$ contradicting its choice. This completes the proof of the first claim of the lemma.

\medskip
To see that $C$ is a longest cycle containing every edge of $S$, it is sufficient to recall that $A\subseteq V(C)$. Then $C$ contains $2(p-s)$ edges $xy$ with $x\in A$ and $y\in B$. Hence, the total number of edges is $2p-s+t$.

We remark that the proof of the lemma can be used to construct the required cycle $C$ in polynomial time. 
\end{proof}

\section{Rerouting Long Cycles to Dense Subgraphs}\label{sec:ext}
In this section, we show that a dense induced subgraph can be used to find a long cycle in a 2-connected graph. Specifically, we show that one can always assume that a long cycle is an extension of a longest cycle in a dense subgraph. To state this more precisely, we need some additional terminology that we introduce next. 

Let $T\subseteq V(G)$ for a graph $G$. A path $P$  is called a \emph{$T$-segment} if $P$ has length at least two, the end-vertices of $P$ lie in $T$, and $v\notin T$ for any internal vertex $v$ of $P$.   A set of internally disjoint paths $\mathcal{P}=\{P_1,\ldots,P_r\}$ is a \emph{system of $T$-segments} if (i) $P_i$ is a $T$-segment for every $i\in\{1,\ldots,r\}$, and (ii) the union of the paths in $\mathcal{P}$ is a linear forest. 
 Let $A,B\subseteq V(G)$ be disjoint sets of vertices in $G$. 
 For a pair $\{x,y\}$ of distinct vertices in $G$, we say that $\{x,y\}$ is an \emph{$A$-pair} (\emph{$B$-pair}, respectively) if $x,y\in A$ ($x,y\in B$, respectively), and we say that $\{x,y\}$ is an \emph{$(A,B)$-pair} if either $x\in A$, $y\in B$ or, symmetrically,  $y\in A$, $x\in B$. 
 If $\{A,B\}$ is a partition of $T\subseteq V(G)$, then for a $T$-segment $P$ with end-vertices $x$ and $y$,  
 $P$ is an \emph{$A$-segment} if $\{x,y\}$ is an $A$-pair, $P$ is a \emph{$B$-segment} if $x,y\in B$, and $P$ is an \emph{$(A,B)$-segment} if $\{x,y\}$ is an $\{A,B\}$-pair.

First, we consider the case when there is a dense subgraph $H$ with the property that for every  potentially cyclable  set  $S$  of at most $k$ pairs of distinct vertices, $H+S$ has a Hamiltonian cycle containing every edge of $S$. We show the following lemma whose proof is almost identical to the proof of Lemma~3 in~\cite{FominGLPSZ20}. Nevertheless, we provide the proof here, as we are proving a slightly different statement, and the proof is useful as a warm-up before the proof of the next more technical lemma.

\begin{lemma}\label{lem:ext-dense}
Let $G$ be a $2$-connected graph and let $k$ be a positive integer. Suppose that $H$ is an induced subgraph of $G$ such that $|V(H)|\geq 2k$ and for every  potentially cyclable  set  $S$  of at most $k$ pairs of distinct vertices of $H$, $H+S$ has a Hamiltonian cycle containing every edge of $S$. Then $G$ has a cycle of length at least $|V(H)|+k$ if and only if one of the following holds:
\begin{itemize}
\item[(i)] There are two distinct vertices $s,t\in V(H)$ such there is an $(s,t)$-path $P$ in $G$ of length at least $k+1$ whose internal vertices lie in $V(G)\setminus V(H)$.
\item[(ii)] There is a system of $T$-segments $\mathcal{P}=\{P_1,\ldots,P_r\}$ for $T=V(H)$ such that $r\leq k$ and the total number of vertices on the paths in $\mathcal{P}$ outside $T$ is at least $k$ and at most $2k-2$.
\end{itemize}
\end{lemma}

\begin{proof}
Let  $T=V(H)$. 
We start with the easier part, where we show that if either (i) or (ii) is fulfilled, then $G$ has a cycle of length at least $|V(H)|+k$.  

Suppose that there are distinct $s,t\in T$ and an $(s,t)$-path $P$ in $G$ with all internal vertices outside $T$ such that the length of $P$ is at least $k+1$.  Let $S=\{st\}$. We have that $H+S$ has a Hamiltonian cycle $C$ containing $st$.  We replace the edge $st$ in $C$ by the path $P$. Then the length of the obtained cycle $C'$ is at least $|V(H)|+k$ as required. 

Suppose that $G$ has a system of  $T$-segments $\mathcal{P}=\{P_1,\ldots,P_r\}$  and the total number of vertices on the paths outside $T$ is at least $k$. 
Let $s_i$ and $t_i$ be the end-vertices of $P_i$ for $i\in\{1,\ldots,r\}$ and define $S=\{s_1t_1,\ldots,s_rt_r\}$. Observe that $S$ is a potentially cyclable set for $H$ and $|S|\leq k$. Then $H+S$ has a Hamiltonian cycle $C$ that contains every edge of $S$. We construct the cycle $C'$ from $C$ by replacing $s_it_i$ by the path $P_i$ for every $i\in\{1,\ldots,r\}$. Because the total number of vertices in the paths of $\mathcal{P}$ outside $T$ is at least $k$, the length of $C'$ is at least $|V(H)|+k$.

\medskip
To show the implication in the other direction, assume that $G$ has a cycle $C$ of length at least $|V(H)|+k$. We consider the following three cases depending on the structure of $C$.

\medskip
\noindent 
{\bf Case~1.} $V(C)\cap T=\emptyset$.
Since $G$ is a 2-connected graph, there are pairwise distinct vertices $s,t\in T$ and $x,y\in V(C)$, and vertex-disjoint  $(s,x)$ and $(y,t)$-paths $P_1$ and $P_2$ such that the internal vertices of the paths are outside $T\cup V(C)$. The cycle $C$ has length at least $|V(H)|+k\geq 3k$. Therefore, $C$
contains an $(x,y)$-path $P$ with at least $k$ vertices. The concatenation of $P_1$, $P$ and $P_2$ is an $(s,t)$-path in $G$ of length at least  $k+1$ whose internal verices  are outside $T$. Hence, (i) holds.

\medskip
\noindent 
{\bf Case~2.} $|V(C)\cap T|=1$.
Let $V(C)\cap T=\{s\}$ for some vertex $s$. Since $G$ is 2-connected, there is an $(x,t)$-path $P$ in $G-s$ such that $x\in V(C)$, $t\in T$ and the internal vertices of $P$ are outside $T\cup V(C)$.
Because the length of $C$ is at least $3k$,  $C$ contains an $(s,x)$-path $P'$ with at least $k+1$ vertices. The concatenation of $P'$ and $P$ is an $(s,t)$-path in $G$ of length at least  $k+1$ whose internal verices  are outside $T$. Hence, (i) holds.

\medskip
\noindent
{\bf Case~3.}  $|V(C)\cap T|\geq 2$. Since $k>0$ and $|V(C)|\geq |V(H)|+k$,  $V(C)\setminus T\neq \emptyset$. 
Then there are pairs of distinct vertices $\{s_1,t_1\}\ldots,\{s_\ell,t_\ell\}$ in $T\cap V(C)$ and paths $P_1,\ldots,P_\ell$ on $C$ such that
(a) $P_i$ is an $(s_i,t_i)$-path for $i\in\{1,\ldots,\ell\}$ with at least one internal vertex and the internal vertices of $P_i$ are outside $T$,  and 
(b) $\bigcup_{i=1}^\ell V(P_i)\setminus T=V(C)\setminus T$. In words, $P_1,\ldots,P_\ell$ form the ``outside'' part of $C$ with respect to $T$. Note that the total number of internal vertices on these paths is at least $k$. 

If there is  $i\in\{1,\ldots,\ell\}$ such that $P_i$ has length at least $k+1$, then  (i) is fulfilled. Assume that this is not the case and the length of each $P_i$ is at most $k$. 
Let $r\in\{1,\ldots,\ell\}$ be the minimum integer such that the total number of internal vertices in $P_1,\ldots,P_r$ is at least $k$. Because each path has at least one internal vertex, $r\leq k$. 
Let $S=\{s_1t_1,\ldots,s_rt_r\}$. By the definition of $S$, these pairs of vertices compose either a linear forest or a cycle. 

Suppose that the pairs in $S$ form a cycle. Then every edge of $C$ is outside $H$, and we have that $r=\ell$ and $C$ is the concatenation of $P_1,\ldots,P_r$. Observe that $r\geq 2$ in this case.  
By the choice of $r$, the total number of internal vertices in $P_1,\ldots,P_{r-1}$ is at most $k-1$. We also have that $P_r$ has at most $k-1$ internal vertices. Because $r\leq k$, $|V(C)|\leq 3k-2$. However, this is a contradiction with $|V(C)|\geq |V(H)|+k\geq 3k$. Therefore, $S$ forms a linear forest. This means that $\mathcal{P}=\{P_1,\dots,P_r\}$ is a system of   $T$-segments for $T=V(H)$ and it holds that $r\leq k$, and the total number of vertices on the paths in $\mathcal{P}$ outside $T$ is at least $k$. To show that (ii) is fulfilled, it remains to prove that the total number of internal vertices on the paths in $\mathcal{P}$ is at most $2k-2$. For this, recall that by the choice of $r$, the total number of internal vertices on $P_1,\ldots,P_{r-1}$ is at most $k-1$. Since the number of internal vertices on $P_r$ is at most $k-1$, the total number of the internal vertices on all paths is at most $2k-2$ as required. This completes the proof.
\end{proof}

Now we show a related result for dense induced subgraphs of another type.
See Figure~\ref{fig:L5} for an illustration.

\begin{figure}[ht]
\centering
\scalebox{0.7}{
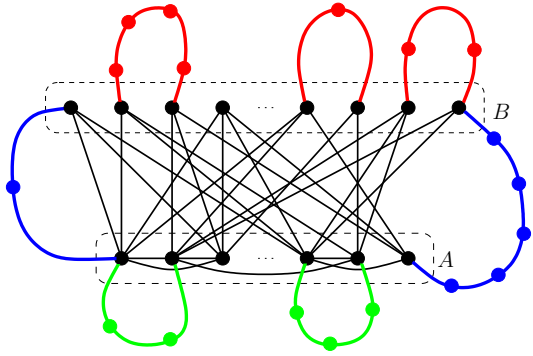}
\caption{Structure of segments in Case (ii) of Lemma~\ref{lem:ext-dense-bip}. The $A$-segments are shown by green lines, the $B$-segments are red, and the $(A,B)$-segments are blue.}
\label{fig:L5}
\end{figure}

\begin{lemma}\label{lem:ext-dense-bip}
Let $G$ be a $2$-connected graph and let $k$ be a positive integer. Suppose that $H$ is an induced subgraph of $G$ whose set of vertices has a partition $\{A,B\}$ with $|A|\geq \frac{3}{2}k$ and $B$ being an independent set. Suppose also that for every potentially cyclable  set  $S$ in $H$ of at most $k$ pairs of distinct vertices in $H$, with $s$ $A$-pairs and $t$ $B$-pairs, $H+S$ has a cycle of length at least $2|A|-s+t$. 
 Then $G$ has a cycle of length at least $2|A|+k$ if and only if one of the following holds:
\begin{itemize}
\item[(i)] There are two distinct vertices $x,y\in V(H)$ such that $H$ has an $(x,y)$-path $P$ of length at least $k+2$ whose internal vertices lie in $V(G)\setminus V(H)$.
\item[(ii)] There is a system of $T$-segments $\mathcal{P}=\{P_1,\ldots,P_r\}$ for $T=V(H)$ with $s$ $A$-segments and $t$ $B$-segments  such that 
\begin{itemize}
\item[(a)] $r\leq k$,
\item[(b)] every $A$-segment has at least two internal vertices,
\item[(c)] the total number of internal vertices on the paths in $\mathcal{P}$ is at least $k+s-t$ and at most $3k-2$.
\end{itemize}
\end{itemize}
 \end{lemma}

\begin{proof}
The proof follows the same lines as the proof of Lemma~\ref{lem:ext-dense} but is more technical.  Let  $T=V(H)$. 
First, we show  that if either (i) or (ii) is fulfilled, then $G$ has a cycle of length at least $|V(H)|+k$. This part is almost identical to the corresponding part of the proof of Lemma~\ref{lem:ext-dense}.
 
Suppose that there are distinct $x,y\in T$ and an $(x,y)$-path $P$ in $G$ with all internal vertices outside $T$ such that the length of $P$ is at least $k+2$.  Let $S=\{xy\}$. We have that $H+S$ has a  cycle $C$ containing $xy$
of length at least $2|A|-1$.  We replace the edge $xy$ in $C$ by the path $P$. Then the length of the obtained cycle $C'$ is at least $2|A|+k$ as required. 

Assume that there is a system of $T$-segments $\mathcal{P}=\{P_1,\ldots,P_r\}$ for $T=V(H)$ with $s$ $A$-segments and $t$ $B$-segments  such that (a)--(b) are fulfilled. 
Let $x_i$ and $y_i$ be the end-vertices of $P_i$ for $i\in\{1,\ldots,r\}$ and define $S=\{x_1y_1,\ldots,x_ry_r\}$. Observe that $S$ is a potentially cyclable set for $H$ and $|S|\leq k$. Then $H+S$ has a cycle $C$ 
of length at least $2|A| $
that contains every edge of $S$. We construct the cycle $C'$ from $C$ by replacing $x_iy_i$ by the path $P_i$ for every $i\in\{1,\ldots,r\}$. Because the total number of internal vertices in the paths of $\mathcal{P}$  is at least $k+s-t$, the length of $C'$ is at least $|V(H)|+k$.

\medskip
For the opposite direction, assume that $G$ has a cycle $C$ of length at least $2|A|+k$. We consider the following three cases. The arguments in the first two cases repeat the arguments in the proof of Lemma~\ref{lem:ext-dense}.

\medskip
\noindent 
{\bf Case~1.} $V(C)\cap T=\emptyset$.
Since $G$ is a 2-connected graph, there are pairwise distinct vertices $x,y\in T$ and $x',y'\in V(C)$, and vertex disjoint  $(x,x')$ and $(y,y')$-paths $P_1$ and $P_2$ such that the internal vertices of the paths are outside $T\cup V(C)$. The cycle $C$ has length at least $2|A|+k\geq 3k$. Therefore, $C$
contains an $(x',y')$-path $P$ with at least $k+1$ vertices. The concatenation of $P_1$, $P$ and $P_2$ is an $(x,y)$-path in $G$ of length at least  $k+2$ whose internal verices  are outside $T$. Hence, (i) is fulfilled.

\medskip
\noindent 
{\bf Case~2.} $|V(C)\cap T|=1$.
Let $V(C)\cap T=\{x\}$ for some vertex $x$. Since $G$ is 2-connected, there is an $(y,y')$-path $P$ in $G-x$ such that $y'\in V(C)$, $y\in T$, and the internal vertices are outside $T\cup V(C)$.
Because the length of $C$ is at least $3k$,  $C$ contains an $(x',y')$-path $P'$ with at least $k+2$ vertices. The concatenation of $P'$ and $P$ is an $(s,t)$-path in $G$ of length at least  $k+2$ whose internal verices  are outside $T$. Hence, (i) holds.

\medskip
\noindent
{\bf Case~3.}  $|V(C)\cap T|\geq 2$. Observe that because $B$ is an independent set, $H$ has no cycle of length greater that $2|A|$. Therefore, as $k>0$ and $|V(C)|\geq 2|A|+k$,  $V(C)\setminus T\neq \emptyset$. 
Let $P_1,\ldots,P_\ell$ be the ``outside'' segments of $C$ with respect to $H$, that is, $P_1,\ldots,P_\ell$ are paths on $C$ such that ($*$) for every $i\in\{1,\ldots,\ell\}$, $P_i$ is an $(x_i,y_i)$-path with at least one internal vertex for some distinct $x_i,y_i\in T$ and the internal verices of $P_i$ are outside $T$, and  ($**$) $\bigcup_{i=1}^\ell V(P_i)\setminus T=V(C)\setminus T$.
If $P_i$ has length at least $k+2$ for some $i\in\{1,\ldots,\ell\}$, then (i) holds. Assume that this is not the case, that is, the length of each $P_i$ is at most $k+1$. 
Let  $I_A,I_B,I_{AB}\subseteq \{1,\ldots,\ell\}$ be the subsets of indices such that $P_i$ is an $B$-segment for $i\in I_B$, a $B$-segment for $i\in I_B$, and an $(A,B)$-segment for $i\in I_{AB}$; note that some of these sets may be empty. 

First, we consider $I_B$. 
Suppose that the paths $P_i$ for $i\in I_B$ have at least $k-|I_B|$ internal vertices. Consider an inclusion minimal subset of indices $J\subseteq I_B$ such that  the paths $P_i$ for $i\in J$ have at least $k-|J|$ internal vertices and let $S=\{x_iy_i\mid i\in J\}$. Observe that the pairs of $S$ compose either a linear forest or a cycle. 
Suppose that the pairs in $S$ form a cycle. Then every edge of $C$ is outside $H$, and we have that  $C$ is the concatenation of the paths $P_i\in J$. Note that $|J|\geq 2$  in this case.
Let $j\in J$. By the choice of $J$, the total number of internal vertices on the paths $P_i$ for $i\in J\setminus\{j\}$ is at most $k-|J|-1$. Because the length of $P_j$ is at most $k+1$, we have that $|V(C)|\leq (k-|J|-1)+|J|+k=2k-1<2|A|+k$; a contradiction. Therefore, $S$ forms a linear forest. We obtain that $\mathcal{P}=\{P_i\mid i\in J\}$ is a system of $T$ segments  and $|\mathcal{P}|\leq k$. To see that the total number of internal vertices on the paths in $\mathcal{P}$ is at most $2k$, let $j\in J$. Because the total number of internal vertices on the paths $P_i$ for $i\in J\setminus\{j\}$ is at most $k-|J|-1$ and the length of $P_j$ is at most $k+1$, the number of internal vertices on the paths in $\mathcal{P}$ is at most $(k-|J|-1)+k\leq 3k-2$. We conclude that (ii) is fulfilled. 

 Assume from now on that the paths $P_i$ for $i\in I_B$ have at most $k-|I_B|-1$ internal vertices. Then we analyse $I_{AB}$ in a similar way.   Let $t=|I_B|$. 
  Suppose that the paths $P_i$ for $i\in I_{AB}\cup I_B$ have at least $k-t$ internal vertices. 
  Consider an inclusion minimal subset of indices $J\subseteq I_{AB}$ such that  the paths $P_i$ for $i\in J\cup I_B$ have at least $k-t$ internal vertices and let $S=\{x_iy_i\mid i\in J\cup I_B\}$. Notice that $|S|\leq k$. 
  Again, we have that  the pairs of $S$ compose either a linear forest or a cycle. Then we exclude the possibility that $S$ forms a cycle. If we have a cycle, 
  then $C$ is the concatenation of the paths $P_i\in J\cup I_B$. Pick an arbitrary  $j\in J$. We have that the total number of internal vertices on the paths $P_i$ for $i\in (J\setminus\{j\})\cup I_B$ is at most $k-t-1$. Because the length of $P_j$ is at most $k+1$,  $|V(C)|\leq (k-t-1)+(|J|+t)+k= 2k+|J|-1<2|A|+k$  and we get a contradiction.  
 Hence, $S$ forms a linear forest and  $\mathcal{P}=\{P_i\mid i\in J\cup I_B\}$ is a system of $T$ segments  and $|\mathcal{P}|\leq k$. To upper bound the total number of internal vertices on the paths in $\mathcal{P}$, let $j\in J$. Because the total number of internal vertices on the paths $P_i$ for $i\in (J\setminus\{j\})\cup I_B$ is at most $k-t-1$ and the length of $P_j$ is at most $k+1$, the number of internal vertices on the paths in $\mathcal{P}$ is at most $2k-t-1\leq 3k-2$. We obtain that (ii) holds.    
 
It remains to consider the case where the paths $P_i$ for $i\in I_{AB}\cup I_B$ have at most $k-t-1$ internal vertices.  For this we analyse $I_A$. Let $I_A'\subseteq I_A$ be the set of indices $i\in I_A$ such that $P_i$ has at least two internal vertices.  Let $r$ be the number of internal vertices on the paths $P_i$ with $i\in I_A'\cup I_B\cup I_{AB}$. Observe that because $B$ is an independent set, $|V(C)|\leq r+t+2|A|-|I_A'|$. Hence, $r+t-|I_A'|\geq k$. We select an inclusion minimal set of indices $J\subseteq I_A'$  such that  the paths $P_i$ for $i\in J\cup I_B\cup I_{AB}$ have at least $k-t+|J|$ internal vertices and let $S=\{x_iy_i\mid i\in J\cup I_B\cup I_{AB}\}$. Let also $s=|J|$.
Observe that because $P_i$ has at least two internal vertices for every $i\in I_A'$, $|S|\leq k$. In the same way as above, the pairs of $S$ compose either a linear forest or a cycle, and we show that it should be a linear forest. If the pairs of $S$ form a  cycle,  then $C$ is the concatenation of the paths $P_i\in J$.  
Let $j\in J$. By the minimality of $J$, the total number of internal vertices on the paths $P_i$ for $i\in (J\setminus\{j\})\cup I_B\cup I_{AB}$ is at most $k+(s-1)-t-1$. Because the length of $P_j$ is at most $k+1$,  $|V(C)|\leq (k+(s-1)-t-1)+(s+t+|I_{AB}|)+k= 2k+|I_{AB}|+2s-2$. Observe that $t+s+|I_{AB}|\leq k$, because if $t+s+|I_{AB}|\geq k+1$, the total number of the internal vertices on the paths $P_i$ for $i\in (J\setminus\{j\})\cup I_B\cup I_{AB}$ would be at least $k+s$. 
Therefore, $|V(C)|\leq 2k+|I_{AB}|+2s-2\leq 4k-2<2|A|+k$; a contradiction. We obtain that  $S$ forms a linear forest and  $\mathcal{P}=\{P_i\mid i\in J\cup I_B\cup I_{AB}\}$ is a system of $T$ segments,
 and $|\mathcal{P}|\leq k$. To get the upper bound for the total number  of internal vertices on the paths in $\mathcal{P}$, let $j\in J$. Because the total number of internal vertices on the paths $P_i$ for $i\in (J\setminus\{j\})\cup I_B$ is at most $k+(s-1)-t-1$ and the length of $P_j$ is at most $k+1$, the number of internal vertices on the paths in $\mathcal{P}$ is at most $2k+s-t-2\leq 3k-2$. We conclude that  (ii) is fulfilled. This concludes the analysis of Case~3 and the proof of the lemma.    
\end{proof}

Fomin et al.~\cite{FominGLPSZ20} proved the following algorithmic result about systems of $T$-segments. 

\begin{proposition}[{\cite[Lemma~4]{FominGLPSZ20}}]\label{prop:syst}
Let $G$ be a graph, $T\subseteq V(G)$, and let $p$ and $r$ be positive integers.  Then it can be decided in $2^{\Oh(p)}\cdot n^{\Oh(1)}$ time whether there is a system of $T$-segments $\mathcal{P}$ with $r$ paths having $p$ internal vertices in total. 
\end{proposition}

However, we need an algorithm for constructing a system of $T$-segments with additional properties described in Lemma~\ref{lem:ext-dense-bip}. For this, we modify the algorithm from Proposition~\ref{prop:syst} (see~\cite[Lemma~4]{FominGLPSZ20}). For simplicity, we show how to solve the decision problem but the algorithm can be easily modified to produce a required system of $T$-segments.   

\begin{lemma}\label{lem:find-syst}
Let $G$ be a graph, $T\subseteq V(G)$, and let $\{A,B\}$ be a partition of $T$. Let also $p$ and $r$ be positive integers, and suppose that $s$ and $t$ are nonnegative integers with $s+t\leq r$.  Then it can be decided in $2^{\Oh(p)}\cdot n^{\Oh(1)}$ time whether there is a system of $T$-segments $\mathcal{P}$ with $r$ paths having $p$ internal vertices in total such that (i) $\mathcal{P}$ contains $s$ $A$-segments, (ii) $t$ $B$-segments, and (iii) every $A$-segment has at least two internal vertices. 
\end{lemma}

\begin{proof}
As Lemma~4 in~\cite{FominGLPSZ20}, our algorithm is based on the \emph{color coding} technique introduced by Alon, Yuster and Zwick
in~\cite{AlonYZ95} (see also~\cite[Chapter~5]{cygan2015parameterized} for the introduction to the technique). Following~\cite{FominGLPSZ20}, we first describe a randomized Monte-Carlo algorithm and then explain how it could be derandomized. 

We say that a system of $T$-segments is \emph{feasible} if it satisfies the conditions of the lemma. 
Notice that if $\mathcal{P}=\{P_1,\ldots,P_r\}$ is a feasible system of $T$-segments, then for the total number of vertices in the paths, we have that $|\cup_{i=1}^rV(P_i)|\leq p+2r$. 
If $r>p$, then  a feasible system of $T$-segments does not exist, because each path in a solution should have at least one internal vertex. 
Hence, we assume without loss of generality that $r\leq p$. Let $q=p+2r\leq 3p$. We color the vertices of $G$ with $q$ colors uniformly at random. Denote by $c\colon V(G)\rightarrow\{1,\ldots,q\}$ the constructed coloring.  
We say that a feasible system of $T$-segments $\mathcal{P}=\{P_1,\ldots,P_r\}$ is \emph{colorful} if the vertices of $\bigcup_{i=1}^rV(P_i)$ are colored by distinct colors. 
We show the following claim.

\begin{claim}\label{cl:alg}
The existence of a colorful feasible system of $T$-segments $\mathcal{P}$ can be verified in $2^{\Oh(p)}\cdot n^{\Oh(1)}$ time.
\end{claim}

\begin{proof}[Proof of Claim~\ref{cl:alg}]
We design a dynamic programming algorithm that decides whether there is a colorful  feasible system of $T$-segments. 

The algorithm works in two stages.  In the first stage, for every two distinct vertices $x,y\in T$ and every set of colors $X\subseteq \{1,\ldots,q\}$ of size at least three, we compute the Boolean function $\alpha(x,y,X)$ such that
$\alpha(x,y,X)=\true$ if and only if there is a $T$-segment $P$ whose end-vertices are $x$ and $y$, $V(P)\cap T=\{x,y\}$, $|V(P)|=|X|$, and the vertices of $P$ are colored by distinct colors from $X$. 
Computing  $\alpha(x,y,X)$ is standard (see~\cite[Chapter~5]{cygan2015parameterized}), because we just find an $(x,y)$-path in $G-(T\setminus\{x,y\})]$ whose vertices are colored by the colors from $X$, and the table of values of the function can be computed in $2^{\Oh(p)}\cdot n^{\Oh(1)}$ time. 

To simplify further computations, we define 
\begin{equation*}
\alpha^*(x,y,X)=
\begin{cases}
\false&\mbox{if }x,y\in A\text{ and }|X|=3,\\
\false&\mbox{if }|X|\leq 2,\\
\alpha(x,y,X)&\mbox{otherwise}. 
\end{cases}
\end{equation*}

In the second stage, for every $x\in T$, all integers $p'$, $r'$, $s'$ and $t'$ such that 
$r'\leq p'\leq p$, $1\leq r'\leq r$, $0\leq s'\leq s$, $0\leq t'\leq t$ and $s'+t'\leq r'$, 
 and every set of colors $X\subseteq \{1,\ldots,q\}$ with $3\leq |X|\leq p'+2r'$, the algorithm computes the value of the Boolean function $\beta(x,p',r',s',t',X)$, where  $\beta(x,p',r',s',t',X)=\true$
 if and only if there is a system of $T$-segments $\mathcal{P}'=\{P_1',\ldots,P_{r'}'\}$ with $r'$ paths having $p'$ internal vertices in total such that 
 \begin{itemize}
 \item[(i)] $\mathcal{P}'$ contains $s'$ $A$-segments and $t'$ $B$-segments, 
 \item[(ii)] every $A$-segment has at least two internal vertices,
 \item[(iii)] for $U=\bigcup_{i=1}^{r'}V(P_i')$, $|U|=|X|$ and the vertices of $U$ are colored by distinct colors from $X$ by the coloring $c$,
 \item[(iv)] $x$  is an end-vertex of exactly one path of $\mathcal{P}'$.  
 \end{itemize}
We are interested in the values of  $\beta(x,p',r',s',t',X)$ for $r'\leq p'\leq p$, $1\leq r'\leq r$, $0\leq s'\leq s$, $0\leq t'\leq t$ and $s'+t'\leq r'$, and $3\leq |X|\leq p'+2r'$, but to simplify computations, we extend the domain and assume that  $\beta(x,p',r',s',t',X)=\false$ if one of these constraints is broken.  Observe that a colorful feasible system of $T$ segments exists if and only if $\beta(x,p,r,s,t,X)=\true$ for some $x\in T$ and $X\subseteq\{1,\ldots,q\}$.

We consecutively compute the tables of values of $\beta(x,p',r',s',t',X)$ for $r'=1,2,\ldots,r$ starting with $r'=1$. For this, we use the computed tables of values of $\alpha(x,y,X)$.

For $r'=1$, by the definition of $\beta(x,p',r',s',t',X)$,
we have that 
\begin{equation}\label{eq:beta-one}
\beta(x,p',r',s',t',X)=
\begin{cases}
\bigvee_{y\in A\setminus\{x\}}\alpha^*(x,y,X)&\mbox{if }x\in A,~s'=1,~t'=0,~p'=|X|-2,\\
\bigvee_{y\in B}\alpha^*(x,y,X)&\mbox{if }x\in A,~s'=0,~t'=0,~p'=|X|-2,\\
\bigvee_{y\in B\setminus\{x\}}\alpha^*(x,y,X)&\mbox{if }x\in B,~s'=0,~t'=1,~p'=|X|-2,\\
\bigvee_{y\in A}\alpha^*(x,y,X)&\mbox{if }x\in B,~s'=0,~t'=0,~p'=|X|-2,\\
\false&\mbox{otherwise};
\end{cases}
\end{equation}
here and further we assume that 
$\bigvee_{z\in Z}\varphi(z)=\false$ for any Boolean function $\varphi(z)$ if $Z=\emptyset$.

For $r'\geq 2$, we use the following recurrences. If $x\in A$, we have 
\begin{align}\label{eq:beta-two-A}
 \beta(x,p',r',s',t',X)=&
 \bigvee_{y\in A\setminus\{x\},Y\subset X}\big(\alpha^*(x,y,Y)\wedge\beta(y,p'-|Y|+2,s'-1,t',(X\setminus Y)\cup\{c(y)\})\big)\nonumber\\
 \vee&\bigvee_{y\in B,Y\subset X}\big(\alpha^*(x,y,Y)\wedge\beta(y,p'-|Y|+2,s',t',(X\setminus Y)\cup\{c(y)\})\big)\nonumber\\
 \vee&\bigvee_{y\in A\setminus\{x\},z\in T\setminus\{x,y\},Y\subset X}\big(\alpha^*(x,y,Y)\wedge\beta(z,p'-|Y|+2,s'-1,t',X\setminus Y)\big)\nonumber\\ 
  \vee&\bigvee_{y\in B,z\in T\setminus\{x,y\},Y\subset X}\big(\alpha^*(x,y,Y)\wedge\beta(z,p'-|Y|+2,s',t',X\setminus Y)\big).
 \end{align}
Symmetrically, if $x\in B$, 
\begin{align}\label{eq:beta-two-B}
 \beta(x,p',r',s',t',X)=&
 \bigvee_{y\in B\setminus\{x\},Y\subset X}\big(\alpha^*(x,y,Y)\wedge\beta(y,p'-|Y|+2,s',t'-1,(X\setminus Y)\cup\{c(y)\})\big)\nonumber\\
 \vee&\bigvee_{y\in A,Y\subset X}\big(\alpha^*(x,y,Y)\wedge\beta(y,p'-|Y|+2,s',t',(X\setminus Y)\cup\{c(y)\})\big)\nonumber\\
 \vee&\bigvee_{y\in B\setminus\{x\},z\in T\setminus\{x,y\},Y\subset X}\big(\alpha^*(x,y,Y)\wedge\beta(z,p'-|Y|+2,s',t'-1,X\setminus Y)\big)\nonumber\\ 
  \vee&\bigvee_{y\in A,z\in T\setminus\{x,y\},Y\subset X}\big(\alpha^*(x,y,Y)\wedge\beta(z,p'-|Y|+2,s',t',X\setminus Y)\big).
 \end{align}

Correctness of (\ref{eq:beta-two-A}) and (\ref{eq:beta-two-B}) is proved by standard arguments. Hence, we only sketch the correctness proof for (\ref{eq:beta-two-A}) (the proof for (\ref{eq:beta-two-B}) is done by the same arguments).

Suppose that $x\in A$ and $\beta(x,p',r',s',t',X)=\true$. By the definition, there is a system of $T$-segments $\mathcal{P}'=\{P_1',\ldots,P_{r'}'\}$ with $r'$ paths having $p'$ internal vertices in total that satisfies conditions (i)--(iv). We assume without loss of generality that $x$ is an end-vertex of $P_1$. Let $y$ be the other end-vertex of $P_1$. Then $\alpha^*(x,y,Y)=\true$. Let also $Y=c^{-1}(V(P_1))$ and 
$\mathcal{P}''=\{P_2',\ldots,P_{r'}'\}$. 
We have four cases depending on whether $y\in A$ or $y\in B$ and on whether $y$ is a shared end-vertex or not. 
Suppose that $y\in A$ and $y$ is an and vertex of another path, say, $P_2$. Then 
$\beta(y,p'-|Y|+2,s'-1,t',(X\setminus Y)\cup\{c(y)\})=\true$. Therefore, the value of the right part of (\ref{eq:beta-two-A}) is $\true$.  The other cases are similar. 
If  $y\in B$ and $y$ is an and vertex of another path, then $\beta(y,p'-|Y|+2,s',t',(X\setminus Y)\cup\{c(y)\})=\true$. 
If  $y\in A$ and $y$ not an end-vertex of $P_2',\ldots,P_{r}'$, then $\beta(z,p'-|Y|+2,s'-1,t',X\setminus Y)=\true$.
If   $y\in B$ and $y$ not an end-vertex of $P_2',\ldots,P_{r}'$, then $\beta(z,p'-|Y|+2,s',t',X\setminus Y)=\true$. 
In all these cases, the value of the right part of (\ref{eq:beta-two-A}) is $\true$.  

For the opposite direction, assume that the value of the right part of (\ref{eq:beta-two-A}) is $\true$. Then either there is $y\in A\setminus\{x\}$ and $Y\subseteq Y$ such that 
$\alpha^*(x,y,Y)\wedge\beta(y,p'-|Y|+2,s'-1,t',(X\setminus Y)\cup\{c(y)\})=\true$, or   there is $y\in B$ and $Y\subseteq Y$ such that
$\alpha^*(x,y,Y)\wedge\beta(y,p'-|Y|+2,s',t',(X\setminus Y)\cup\{c(y)\})=\true$, or 
there  are $y\in A\setminus\{x\}$, $z\in T\setminus \{x,y\}$, and $Y\subseteq Y$ such that $\alpha^*(x,y,Y)\wedge\beta(z,p'-|Y|+2,s'-1,t',X\setminus Y)=\true$, or 
there  are $y\in B$, $z\in T\setminus \{x,y\}$, and $Y\subseteq Y$ such that $\alpha^*(x,y,Y)\wedge\beta(z,p'-|Y|+2,s',t',X\setminus Y)=\true$. 
The arguments for these four cases are very similar. Therefore, we consider only the first case when there is $y\in A\setminus\{x\}$ and $Y\subseteq Y$ such that 
$\alpha^*(x,y,Y)\wedge\beta(y,p'-|Y|+2,s'-1,t',(X\setminus Y)\cup\{c(y)\})=\true$. 

Because $\alpha^*(x,y,Y)=\true$, $G$ has an $(x,y)$-path $P$ with $|Y|-2\geq 2$ internal vertices and the vertices of $P$ are colored by distinct colors from $X$.  Note that $P$ is an $A$-segment. Since $\beta(y,p'-|Y|+2,s'-1,t',(X\setminus Y)\cup\{c(y)\})=\true$, we have that is a system of $T$-segments $\mathcal{P}''=\{P_1',\ldots,P_{r'-1}'\}$ with $r'-1$ paths having $p'-|Y|+2$ internal vertices such that 
(i$^*$)$\mathcal{P}''$ contains $s'-1$ $A$-segments and $t'$ $B$-segments, 
 (ii$^*$) every $A$-segment has at least two internal vertices,
 (iii$^*$) for $U=\bigcup_{i=1}^{r'-1}V(P_i')$, $|U|=|(X\setminus Y)\cup \{c(y)\}|$ and the vertices of $U$ are colored by distinct colors from $(X\setminus Y)\cup \{c(y)\}$ by the coloring $c$, and 
(iv$^*$) $y$  is an end-vertex of exactly one path of $\mathcal{P}''$.   Let $\mathcal{P'}=\{P,P_1',\ldots,P_{r'-1}'\}$. Then $\mathcal{P}'$ is a system of $T$-segments with $r'$ path having $p'$ internal vertices and conditions (i)--(iv) are fulfilled. Therefore, $\beta(x,p',r',s',t',X)=\true$. This completes the correction proof.

To evaluate the running time, note that the values of $\beta(x,p',r',s',t',X)$ are computed for at most $n$ vertices $x$, at most $n^4$ $4$-tuples of integers $p',r',s',t'$, and at most $2^q=2^{\Oh(p)}$ sets $X$.  
Because the table of values of $\alpha(x,y,X)$ can be computed in $2^{\Oh(p)}\cdot n^{\Oh(1)}$ time, the initial table of values of $\beta(x,p',r',s',t',X)$ for $r'=1$ is computed in $2^{\Oh(p)}\cdot n^{\Oh(1)}$ time. To compute the value of $\beta(x,p',r',s',t',X)$, we use either (\ref{eq:beta-two-A}) or  (\ref{eq:beta-two-A}). In these recurrences, we go through at most $n$ choices of $y$ and $z$, and consider ay most $2^q$ subsets $Y$. This means, that for each $r'\geq 2$, $\beta(x,p',r',s',t',X)$ is computed in $2^{\Oh(p)}\cdot n^{\Oh(1)}$ from the previously computed tables. We conclude that overall running time is  
$2^{\Oh(p)}\cdot n^{\Oh(1)}$. This concludes the proof of the claim.
\end{proof}

Assume that there is a feasible system of $T$-segments. We upper bound the probability that there is no colorful system. Let $\mathcal{P}=\{P_1,\ldots,P_r\}$ be a  feasible system of $T$-segments. Since the paths of $\mathcal{P}$ have at most $q$ vertices in total, the probability that the vertices of paths are colored by distinct colors if we assign the colors uniformly at random is at least $\frac{q!}{q^q}\leq e^{-q}\leq e^{-3p}$. Then the probability that there are two vertices with the same colors is at  most $1-e^{-3p}$. 

This observation leads us to a Monte-Carlo algorithm. We consequently construct at most $e^{3p}$ random colorings $c\colon V(G)\rightarrow\{1,\ldots,q\} $. For each coloring, we use Claim~\ref{cl:alg}
to verify whether there is  a colorful feasible system of $T$-segments $\mathcal{P}$. If we find such a system, we return the yes answer and stop. Otherwise, if we fail to find a colorful system for $e^{3p}$ random colorings, we return the no answer. The probability that this negative answer is false is at most $(1-e^{-3e})^{3e}\leq e^{-1}<1$. This means, that the probability of the false negative answer is upper bounded by a constant $e^{-1}<1$, 
The running time of the algorithm is $2^{\Oh(p)}\cdot n^{\Oh(1)}$.  

This algorithm can be derandomized using standard tools (see~\cite{AlonYZ95} and \cite[Chapter~5]{cygan2015parameterized}). This is done by using \emph{perfect hash functions} (we refer to~\cite[Chapter~5]{cygan2015parameterized} for the definition) instead of random colorings.  The currently best explicit construction of such families was done by Naor, Schulman and Srinivasan in~\cite{NaorSS95}. 
The family of perfect hash function in our case has size $e^{3p}\cdot p^{O(\log p)}\cdot \log n$ and can be constructed in time $e^{3p}\cdot p^{O(\log p)}\cdot n\log n$~\cite{NaorSS95}. 
This allows to obtain a deterministic algorithm that runs in $2^{\Oh(p)}\cdot n^{\Oh(1)}$ time. 
\end{proof}

\section{Proof of the Main Result}\label{sec:main}
Now we have all ingredients to prove our main result. We restate it here for the reader's convenience.

\main*

\begin{proof}
Let $(G,k)$ be an instance of  \probLCMAD, where $G$ is a 2-connected graph. We use the algorithm from Proposition~\ref{prop:densest} and compute $\mad(G)$ in polynomial time. If $k=0$, the problem is trivial, because a cycle of length at least $\mad(G)$ exists by Theorem~1. Hence, we can assume that $k\geq 1$. If $k>\frac{1}{88}\mad(G)-1$, 
we use  Proposition~\ref{prop:LC-best} and solve the problem in $2^{\Oh(k)}\cdot n^{\Oh(1)}$ time. From now, we assume that $0<k\leq \frac{1}{88}\mad(G)-1$. In particular,  $k\leq \frac{1}{80}\mad(G)-1$.
We apply Lemma~\ref{lem:fid-dense}, and in polynomial time either
\begin{itemize}
 \item[(i)]  find a cycle of length at least $\mad(G)+k$ in $G$, or
 \item[(ii)] find an induced subgraph $H$ of $G$ with $\ad(H)\geq \mad(G)-1$ such that  $\delta(H)\geq \frac{1}{2}\ad(H)$ and $|V(H)|< \ad(H)+k+1$, or
 \item[(iii)] find an induced subgraph $H$ of $G$ 
 such that there is a partition $\{A,B\}$ of $V(H)$ with the following properties:
 \begin{itemize}
 \item $B$ is an independent set,
 \item $\frac{1}{2}\mad(G)-4k\leq |A|$,
 \item for every $v\in A$, $|N_H(v)\cap B|\geq 2|A|$,
 \item for every $v\in B$, $\dg_H(v)\geq |A|-2k-2$.
 \end{itemize}
  \end{itemize}

If the algorithm finds a cycle of length at least $\mad(G)+k$, then we return it and stop. In Cases~(ii) and (iii), we get a dense induced subgraph $H$ that can be used to find a solution if it exists.

\medskip
\noindent
{\bf Case~(ii).} The algorithm from Lemma~\ref{lem:fid-dense} returns an induced subgraph $H$ of $G$ with $\ad(H)\geq \mad(G)-1$ such that  $\delta(H)\geq \frac{1}{2}\ad(H)$ and $|V(H)|< \ad(H)+k+1$. 
Let $k'=\lceil\mad(G)\rceil+k-|V(H)|$. We have that $G$ has a cycle of length at least $\mad(G)+k$ if and only if $G$ has a cycle of length at least $|V(H)|+k'$. 
Note that $k'\leq k+1 \leq \frac{1}{88}\mad(G)\leq \frac{1}{60}\ad(H)$. By Lemma~\ref{lem:dense}, for every potentially cyclable set $S$ of at most $k+1$ pairs of distinct vertices of $H$,  $H+S$ has a Hamiltonian cycle containing every edge of $S$. 

Suppose that $k'\leq 0$. Observe that $H$ has a Hamiltonian cycle as we can use Lemma~\ref{lem:dense} for $S=\{e\}$, where $e$ is an arbitrary edge $e\in E(H)$. Then we conclude that $H$ has a cycle of length at least $\mad(G)+k$ and stop. Assume that $k'>0$. Note that $|V(H)|\geq\ad(H)\geq k'$. Then by Lemma~\ref{lem:ext-dense},  $G$ has a cycle of length at least $|V(H)|+k'$ if and only if one of the following holds:
\begin{itemize}
\item[(a)] There are two distinct vertices $s,t\in V(H)$ such that $H$ has an $(s,t)$-path $P$ of length at least $k'+1$ whose internal vertices  lie in $V(G)\setminus V(H)$.
\item[(b)] There is a system of $T$-segments $\mathcal{P}=\{P_1,\ldots,P_r\}$ for $T=V(H)$ such that $r\leq k'$ and the total number of vertices on the paths in $\mathcal{P}$ outside $T$ is at least $k'$ and at most $2k'-2$.
\end{itemize}

First, we check if (a) can be satisfied. For this, we consider all pairs of distinct vertices $s$ and $t$ of $H$. For every pair, we construct $G'=G[(V(G)\setminus V(H))\cup\{s,t\}]$ and use Proposition~\ref{prop:st} to find an $(s,t)$-path of length at least $k'+1$ in $G$ in $2^{\Oh(k)}\cdot n^{\Oh(1)}$ time. If we find such a path for some pair, we report the existence of a cycle of length at least $\mad(G)+k$ and stop. Otherwise, we verify (b) using Proposition~\ref{prop:syst}. We use the algorithm from Proposition~\ref{prop:syst} for $r\in\{1,\ldots,k'\}$ and for $p\in\{k',\ldots,2k'-2\}$. If we find a required system of $T$-segments, then we return that $G$ has a cycle of of length at least $\mad(G)+k$ and stop. If we fail to find such a system for every $r$ and $p$, we conclude that $G$ has no cycle of length at least $\mad(G)+k$. Note that this can be done in $2^{\Oh(k)}\cdot n^{\Oh(1)}$ time. This concludes Case~(ii).

\medskip
\noindent
{\bf Case~(iii).} The algorithm from Lemma~\ref{lem:fid-dense} returns an induced subgraph $H$ of $G$  such that there is a partition $\{A,B\}$ of $V(H)$ with the properties:
 \begin{itemize}
 \item $B$ is an independent set,
 \item $\frac{1}{2}\mad(G)-4k\leq |A|$,
 \item for every $v\in A$, $|N_H(v)\cap B|\geq 2|A|$,
 \item for every $v\in B$, $\dg_H(v)\geq |A|-2k-2$.
\end{itemize}
Let $k'=\lceil\mad(G)\rceil+k-2|A|$. Observe that $G$ has a cycle of length at least $\mad(G)+k$ if and only if $G$ has a cycle of length at least $2|A|+k'$. 
We have that $2|A|\geq \lceil\mad(G)\rceil-8k$ and, therefore, $k'\leq 9k$. 

Note that $|A|\geq \frac{1}{2}\mad(G)-4k\geq 40k$, since $k\leq \frac{1}{88}\mad(G)-1$. Also, we have that for every $v\in B$, $\dg_H(v)\geq |A|-4k$.   Therefore, by Lemma~\ref{lem:dense-bip}, for every potentially cyclable set $S$ of at most $9k$ pairs of distinct vertices,  $G'=G+S$ has a cycle $C$ containing every edge of $S$ and the length of $C$ is $2|A|-s+t$, where $s$ in the number of edges of $S$ with both end-vertices in $A$ and $t$ is the number of edges in $S$ with both end-vertices in $B$. 
 
Suppose that $k'\leq 0$. Then we observe that $H$ has a cycle of length $2|A|$ because we can set $S=\{xy\}$, where $xy\in E(H)$ with $x\in A$ and $y\in B$. Then $H$ has a cycle of length at least $2|A|+k'$ and we conclude that $G$ has a cycle of length at least $\mad(G)+k$. Assume that $k'>0$. Since  $|A|\geq 40k\geq \frac{3}{2}k'$, we can apply Lemma~\ref{lem:ext-dense-bip}. We obtain that  
  $G$ has a cycle of length at least $2|A|+k'$ if and only if one of the following holds:
\begin{itemize}
\item[(a)] There are two distinct vertices $x,y\in V(H)$ such that $H$ has an $(x,y)$-path $P$ of length at least $k'+2$ whose internal vertices  are in $V(G)\setminus V(H)$.
\item[(b)] There is a system of $T$-segments $\mathcal{P}=\{P_1,\ldots,P_r\}$ for $T=V(H)$ with $s$ $A$-segments and $t$ $B$-segments  such that 
\begin{itemize}
\item $r\leq k'\leq 9k$,
\item every $A$-segment has at least two internal vertices,
\item the total number of internal vertices vertices on the paths in $\mathcal{P}$ is at least $k'+s-t$ and at most $3k'-2\leq 27k-2$.
\end{itemize}
\end{itemize}

To verify (a), we use the same approach as in Case~(ii), that is,  we consider all pairs of distinct vertices $x$ and $y$ of $H$. For every pair, we construct $G'=G[(V(G)\setminus V(H))\cup\{x,y\}]$ and use Proposition~\ref{prop:st} to find an $(x,y)$-path of length at least $k'+2$ in $G$ in $2^{\Oh(k)}\cdot n^{\Oh(1)}$ time. If we find such a path for some pair, we report the existence of a cycle of length at least $\mad(G)+k$ and stop. Otherwise, we verify (b) using Lemma~\ref{lem:find-syst}. We use the algorithm from this lemma for $r\in\{1,\ldots,k'\}$,  $s\in\{0,\ldots,k'\}$ and $t\in\{0,\ldots,k'\}$ such that $s+t\leq r$, and for $p\in\{k'+s-t,3k'-2\}$. 
If we find a system of $T$-segments $\mathcal{P}=\{P_1,\ldots,P_r\}$ for $T=V(H)$ with $s$ $A$-segments and $t$ $B$-segments with the required properties, then we conclude that $G$ has a cycle of of length at least $2|A|+k'$ and stop. If such a system does not exist for every choice of $r$, $s$, $t$, and $p$, we have that $G$ has no cycle of length at least $\mad(G)+k$. 
By Lemma~\ref{lem:find-syst}, this can be done in $2^{\Oh(k)}\cdot n^{\Oh(1)}$ time, because $k'\leq 9k$. 
This concludes Case~(iii).

Because the algorithm from Lemma~\ref{lem:fid-dense} is polynomial and the other subroutines used in our algorithm for \probLCMAD run in $2^{\Oh(k)}\cdot n^{\Oh(1)}$, the overall running time is  $2^{\Oh(k)}\cdot n^{\Oh(1)}$ and this concludes the proof.

Let us remark that since the algorithms for paths in Propositions~\ref{prop:syst} and \ref{prop:st} and Lemma~\ref{lem:find-syst} are, in fact, constructive, and the same holds for the algorithms for cycles 
in Lemmas ~\ref{lem:ext-dense}  and \ref{lem:ext-dense-bip} and Proposition~\ref{prop:LC-best}, our algorithm is not only able to solve the decision problem, but also can find a cycle of length at least $\mad(G)+k$ if it exists.
\end{proof}  
  
 We observe that the $2$-connectivity condition in Theorem~\ref{thm:mad-fpt}  is crucial for tractability and we cannot drop it even if we consider the problem of finding a cycle whose length exceeds the bound $\eg(G)$ of  Erd{\H{o}}s and Gallai by one.
   
 \lowerbound*  
 
 \begin{proof} 
  We demonstrate an easy reduction from the \HamCycle problem that is \classNP-complete (see~\cite{GareyJ79}). Let $G$ be a graph with $n\geq 3$ vertices and $m$ edges.
  We also assume without loss of generality that $\eg(G)\leq n-1$; otherwise, $G$ is Hamiltonian by Theorem~\ref{thm:EG}. For every vertex $v\in V(G)$, we construct a clique $X_v$ with $n-2$ vertices and then make the vertices of $X_v$ adjacent to $v$. Denote by $G'$ the obtained graph. Then $n'=|V(G')|=n(n-1)$ and $m'=|E(G')|=n\binom{n-1}{2}+m$. We have that
  \begin{equation*} 
  \eg(G')=\frac{2m'}{n'-1}=\frac{n(n-1)(n-2)+m}{n(n-1)-1}>n-2.
  \end{equation*}
  Becase $\eg(G)\leq n-1$, $2m\leq (n-1)^2$ and 
    \begin{equation*} 
  \eg(G')=\frac{2m'}{n'-1}=\frac{n(n-1)(n-2)+m}{n(n-1)-1}\leq n-1.
  \end{equation*}  
  Then $G'$ has a cycle of length at least $\eg(G')+1$ if and only if it has a cycle of length at least  $n$. By the construction of $G'$, $G'$ has a cycle of length at least $n$ if and only if $G$ has such a cycle, that is, if and only if $G$ is Hamiltonian.
 \end{proof}

\end{document}